\documentclass[sn-mathphys]{sn-jnl}

\jyear{2023}%

\theoremstyle{thmstyleone}%
\newtheorem{theo}{Theorem}
\newtheorem{prop}[theo]{Proposition}%
\newtheorem{lemm}[theo]{Lemma}
\newtheorem{coro}[theo]{Corollary}

\theoremstyle{thmstyletwo}%
\newtheorem{example}{Example}%
\newtheorem{remark}{Remark}%

\theoremstyle{thmstylethree}%
\newtheorem{definition}{Definition}%

\newcommand\cA{{\cal A}}
\newcommand\cI{{\cal I}}

\newcommand\cF{{\cal F}}

\newcommand\cL{{\cal L}}
\newcommand\cB{{\cal B}}
\newcommand\cH{{\cal H}}

\newcommand\cO{{\cal O}}
\newcommand\cP{{\cal P}}

\newcommand\cS{{\cal S}}
\newcommand\cV{{\cal V}}
\newcommand\cT{{\cal T}}

\newcommand\Reel{{\bf R}}

\newcommand\N{\mathbb{N}}

\newcommand\ep{{\varepsilon}}

\DeclareMathOperator{\essinf}{ess\;inf}
\DeclareMathOperator{\esssup}{ess\;sup}

\newcommand\bP{{\bf {\rm P}}}

\newcommand\bea{\begin{eqnarray}}
\newcommand\eea{\end{eqnarray}}
\newcommand\bean{\begin{eqnarray*}}
\newcommand\eean{\end{eqnarray*}}

\begin{document}

\title[No-arbitrage conditions and pricing   from discrete-time to  continuous-time strategies]{No-arbitrage conditions and pricing   from discrete-time to  continuous-time strategies}


\author[1,3]{\fnm{Dorsaf} \sur{CHERIF}}\email{dorsaf-cherif@hotmail.fr}

\author[2,3]{\fnm{Emmanuel} \sur{LEPINETTE}}\email{emmanuel.lepinette@ceremade.dauphine.fr}
\equalcont{These authors contributed equally to this work.}

\affil[1]{\orgdiv{Latao}, \orgname{Faculty of Sciences of Tunis, El Manar University}, \orgaddress{\street{94-Rommana}, \city{Tunis}, \postcode{1068},  \country{Tunisia}}}

\affil[2]{\orgdiv{Ceremade, CNRS}, \orgname{Paris Dauphine University, PSL}, \orgaddress{\street{Place du Mar\'echal De Lattre De Tassigny}, \city{Paris}, \postcode{Cedex 16}, \country{France}}}

\affil[3]{\orgdiv{Gosaef}, \orgname{Faculty of Sciences of Tunis, El Manar University}, \orgaddress{\street{94-Rommana}, \city{Tunis}, \postcode{1068}, \country{Tunisia}}}

\abstract{In this paper, a general framework is developed for  continuous-time financial market models defined from  simple strategies through conditional topologies that avoid stochastic calculus and do not necessitate  semimartingale models. We then compare the usual no-arbitrage conditions of the literature, e.g. the usual no-arbitrage conditions NFL, NFLVR and NUPBR  and the recent AIP condition. With  appropriate pseudo-distance topologies, we show that they hold in continuous time if and only if they hold in discrete time. Moreover, the super-hedging prices in continuous time coincide with the discrete-time super-hedging prices, even without any no-arbitrage condition. 
}




\maketitle

\section{Introduction}\label{sec1}
Absence of arbitrage opportunities is an usual condition imposed on financial market models to deduce a characterization of super-hedging prices. In continuous-time, Delbaen and Schachermayer \cite{DScha1} have introduced the famous no-arbitrage condition NFLVR  as equivalent to the existence of a local martingale measure, see also the well known NFL condition by Kreps \cite{Kreps} at the origin of the arbitrage theory in continuous time. More recently, the weaker NUPBR no-arbitrage condition \cite{KK} has been introduced as the minimal one necessary to solve utility maximization problems.  \smallskip

However, models where the price processes are not semi-martingales are also considered in the literature, e.g. fractional Brownian motion, see \cite{Pakk} and \cite{Lo} for empirical studies. Moreover, in the papers  \cite{Rogers} and \cite{Sot}, it is shown that arbitrage opportunities exist in 
fractional Brownian motion models. Also Guasoni considers \cite{Gua1}  non-semimartingale models with transaction costs. In the paper \cite{CL}, the  no-arbitrage condition AIP ensures the finiteness of the super-hedging prices in non-semimartingale frictionless models and a dynamic programming principle allows to compute them in discrete time. \smallskip

Absence of arbitrage opportunities in non-semimartingale models has also been considered by restricting the class of admissible trading strategies as initiated by \cite{Ch}, \cite{BSV}, \cite{BS}, \cite{Sayit} among others. Precisely, only simple strategies with a minimal deterministic time between two trades are allowed. It is then possible to show that fractional Brownian motions, and more general processes, are arbitrage free with respect to this so-called Cheridito's class of simple strategies, see  \cite{JPS}. In other words, this specific restricted class of simple strategies is adapted to the  non-semimartingale price processes of consideration in such a way that a no-arbitrage condition holds.\smallskip

Our approach is different: We fix an a priori given class of strategies that are interpreted as simple discrete-time strategies (discrete-time or simple strategies in short) and the continuous-time strategies are defined as convergent sequences of simple strategies. Here, convergence should be understood with respect to a topology induced by a (conditional) pseudo-distance we introduce in such a way that, by definition, a terminal continuous-time portfolio value is attainable from a terminal discrete-time portfolio process, up to an arbitrarily small error. Precisely, if $\overline  v_{T}$ is a  terminal  continuous-time portfolio value, then for every  $\ep>0$, there exists a terminal discrete-time portfolio value $v_{T}$ such that $v_{T}\ge \overline  v_{T} - \ep $.\smallskip

We aim to show that the usual no-arbitrage conditions NFL, NFLVR and NUPBR in discrete-time are respectively equivalent to their analogous conditions in continuous time, with an appropriate choice of a pseudo-distance topology which is financially meaning. The same holds for the weaker AIP condition which means that non negative payoffs admit non negative prices, or equivalently, the infimum super-hedging price of a non negative price cannot be $-\infty$, see \cite{CL}. Moreover, we then show that the infimum super-hedging prices in discrete time and in continuous time coincide, without supposing any no-arbitrage condition. Of course, such prices may be numerically estimated only if AIP holds, which is the weaker no-arbitrage condition of consideration.\smallskip

In the following, we first present the general framework that generates the continuous-time portfolios from the discrete-time ones without any semi-martingale setting. Then, we successively compare in discrete time and in continuous time the NFL, NFLVR, AIP and NUPBR no-arbitrage conditions. Finally, we compare the super-hedging prices in discrete time and in continuous time. The last section exposes the theory we have developed on pseudo-distance topologies. In the appendix, some auxiliary results are collected.

\section{Model}\label{sec2}

Let $(\Omega,(\cF_t)_{t\in [0,T]},\bP)$ be a complete stochastic basis which is right-continuous. We consider a financial market model defined by $d$ risky assets described by a continuous-time right-continuous price process $S_t=(S_t^1,...,S_t^d) \in \Reel^d_+$, $t\in [0,T]$, adapted to the filtration $(\cF_t)_{t\in [0,T]}$. Moreover, we suppose that there exists a bond whose price is $S^0=1$, without loss of generality. The quantities invested in a portfolio are described, as usual, by a real-valued adapted process $\theta^0$ that describes the quantity invested in the bond and an adapted process  $\theta=( \theta^1 ,..., \theta^d)\in \Reel^d$,  called  strategy, that describes the quantities invested in the risky assets. Without transaction costs, the liquation value of the strategy $\theta$ is  given by the portfolio process $V=V^\theta=\theta S $ where the product needs to be understood as the Euclidean inner product on $\Reel^d$. Recall that, in discrete-time $t=0,1,\cdots,T$, $V=V^\theta$ is said self-financing if  $\theta_{t} S_{t+1}= \theta_{t+1}S_{t+1}$, i.e. $\Delta V_{t+1}:=V_{t+1}-V_t=\theta_t \Delta S_{t+1}$. Then, the terminal value of a self-financing portfolio process starting from the zero initial capital is of the form  $V_{t,T}= \sum \limits_{u=t}^T  \theta_{u-1} \Delta S_u$.

In the following, $T>0$ is the horizon time and we consider for any time $t\le T$,   a set $\cV_{t,T}$  of $T$- terminal discrete-time portfolios, starting from the  zero initial capital at time $ t$. An element of $\cV_{t,T}$ may be seen as a portfolio value generated by a simple strategy, as in \cite{Ch} or generated by specific discrete-time strategies more generally. \smallskip

A first  typical example is when the trades are only executed at arbitrary deterministic times:\bea \label{Elmt1}\cV_{t,T}^{{\rm det}}= \left\{\sum \limits_{i=1}^n  \theta_{t_{i-1}} \Delta S_{t_i},\, t=t_0<\cdots<t_n=T, \,\theta_{t_i}\in L^0( \mathbb{R}^d,\cF_{t_i}),\,n\ge 1 \right\}. \quad\eea
A second example is when the  portfolios are revised at some stopping times, e.g. when some market conditions are satisfied. Let us denote by $\cT_{t,T}$ the set of all $[t,T]$-valued stopping times. We denote by $\hat{\cT}_{t,T}^n$, $n\ge 1$, the set of all increasing sequences of stopping times $(\tau_{i})_{i=0}^n$ such that $t=\tau_0<\cdots<\tau_n=T$. We then consider the set:
\bea \label{Elmt2}\cV_{t,T}^{\rm rand}= \left\{\sum \limits_{i=1}^n  \theta_{\tau_{i-1}} \Delta S_{\tau_i},\, (\tau_{i})_{i=0}^n\in \hat{\cT}_{t,T}^n,\, \theta_{\tau_i}\in L^0( \Reel^d,\cF_{\tau_i}),\,n\ge 1 \right\}.\quad \eea

\begin{remark}\label{Operator}{\rm 
 In the common cases, the discrete-time portfolio processes $V_{t,T}\in \cV_{t,T}$ are explicitly characterized by a priori given "simple" strategies $\theta^{t,T}\in \cS_{t,T}$, i.e. $V_{t,T}=\cI(\theta^{t,T})$ for some operator $\cI$. In that case, we also denote by $V_{t,u}$ the $u$-time value of $V_{t,T}$, i.e.  $V_{t,u}=\cI(\theta^{t,T,u})$, $u\in [t,T]$, where $\theta^{t,T,u}$ is the restriction of $\theta^{t,T}$ to the interval $[t,u]$ so that  $\theta^{t,T,u}_v=0$ if $v>u$.   This is the case in the two examples above and we write $\cV_{t,T}=\cI(\cS_{t,T})$. In continuous-time, this is usual to require the strategies to be admissible. In the example given by (\ref{Elmt2}), we have
 \bea \label{IntegOperator} \cI_u(\theta):=\cI(\theta^{t,T,u})=\sum \limits_{i=1}^n  \theta_{\tau_{i-1}} \left( S_{\tau_i\wedge u}-S_{\tau_{i-1}\wedge u}\right),\quad u\in [t,T].\eea
 We say that $\theta$ is admissible if there exists $m\in \Reel$ such that $\cI_u(\theta)\ge m$ a.s. for all $u\in [t,T]$. In that case, the corresponding set of terminal portfolio processes is denoted by $^{a}\cV_{t,T}$ instead of $\cV_{t,T}$. }$\Diamond$
 
  \end{remark}\smallskip

In the following, we consider $L^0(\Reel^d,\cF_T)$, $d\ge 1$,  the set of all equivalence classes of random variables defined on $(\Omega,\cF_T,\bP)$ with values in $\Reel^d$. The following definitions allow to define continuous-time portfolio processes (resp. strategies) from discrete-time portfolio processes (resp. simple strategies).

\begin{definition} Let $t\le T$ and let $\cO_t$ be a topology on  $L^0(\Reel,\cF_T)$. We say that a sequence $(V_{t,T}^n )_{n\ge 1}$ of $\cV_{t,T}$ is $\cO_t$-integrable if $(V_{t,T}^n )_{n\ge 1}$ is convergent with respect to $\cO_t$.
\end{definition}

The definition above is designed for an arbitrary topology $\cO_t$. It will be used for the particular topology $\cO_t$ as defined in Section \ref{AIP-continuous} below.

\begin{definition} Let $t\le T$ and let $\cO_t$ be a topology on $L^0(\R,\cF_T)$. We denote by $\cV_{t,T}^c=\cV_{t,T}^c(\cO_t)$ the family of all limits for the topology  $\cO_t$ of  $\cO_t$-integrable sequences $(V_{t,T}^n )_{n\ge 1}$ of $\cV_{t,T}$. An element of $\cV_{t,T}^c$ is called a terminal continuous-time portfolio.
 \end{definition}

\begin{definition} Let $t\le T$ and let $\cO_t$ be a topology on  $L^0(\R,\cF_T)$. Suppose that $\cV_{t,T}=\cI(\cS_{t,T})$ for some operator $\cI$ and simple strategies $\cS_{t,T}$.     We say that a sequence $(\theta^n )_{n\ge 1}$ of $\cS_{t,T}$ is $\cO_t$-integrable if $(V_{t,u}^n=\cI_u (\theta^n))_{n\ge 1}$ is $\cO_t$-integrable for all $u\le T$.
\end{definition}

\begin{definition} Let $t\le T$ and let $\cO_t$ be a topology on  $L^0(\R,\cF_T)$. Suppose that $\cV_{t,T}=\cI(\cS_{t,T})$ for some operator $\cI$ and simple strategies $\cS_{t,T}$. A continuous-time strategy $\theta$ on $[t,T]$ is an $\cO_t$-integrable sequence $\theta=(\theta^n)_{n\ge 1}$ of simple strategies $\theta^n\in \cS_{t,T}$. In that case, for any $u\in [t,T]$, we define $V_{t,T}^c(u)=\cI_u(\theta)$ as a limit in $\cO_t$  of the convergent sequence $(\cI_u(\theta^n))_{n\ge 1}$. We then have $V_{t,T}^c=V_{t,T}^c(T)\in \cV_{t,T}^c=\cV_{t,T}^c(\cO_t)$ by definition.  \end{definition}

The aim of the paper is to understand whether a no-arbitrage condition imposed on the set of all discrete-time portfolio processes (or simple strategies)  at any time $t$ also holds on the set of all continuous-time portfolio processes (resp. strategies). Clearly, that should depend on the topologies $(\cO_t)_{t\in [0,T]}$. Also, it is interesting to compare the super-hedging prices obtained by the discrete-time portfolio processes from the continuous-time ones. \smallskip

In the following, we shall consider at any time $t\le T$ a  topology  $\cO_t$ that satisfies the Fatou property defined as follows:

\begin{definition} 
A topology $\cO_t$ on $L^0(\R,\cF_T)$ satisfies the Fatou property if for any sequence $(X^n)_{n\ge 1}$ of  $L^0(\R,\cF_T)$ that converges to $X$ in $\cO_t$, we have $X\le \liminf_{n}X_{k_n}$ for some subsequence $(k_n)_{n\ge 1}$.

\end{definition}

Note that the Fatou property holds as soon as $X= \liminf_{n}X_{k_n}$ for some subsequence $(k_n)_{n\ge 1}$. This is the case for the usual topologies, in particular the topologies defined with respect to the convergence in probability or the $L^p$ norms $\|X\|_p=(E\vert X \vert ^p)^{1/p}$, $p\in [1,\infty]$. We shall see that this is also the case for the topology of Section \ref{AIP-continuous}. This is a non Hausdorff topology which satisfies the following properties:

\begin{definition} 
A topology $\cO_t$ on $L^0(\R,\cF_T)$ is said $\cF_t$-positively homogeneous if for any sequence $(X^n)_{n\ge 1}$ of  $L^0(\R,\cF_T)$ that converges to $X$ in $\cO_t$, and for all $\alpha_t\in L^0(\Reel^+,\cF_t)$, $(\alpha_t X^n)_{n\ge 1}$ converges to $\alpha_t X$ in $\cO_t$.
\end{definition}

\begin{definition} 
A topology $\cO_t$ on $L^0(\R,\cF_T)$ is said $\cF_t$-lower bond preserving if, for any $X\in L^0(\R,\cF_T)$ such that $X\ge m_t$ for some $m_t\in L^0(\R,\cF_t)$ and    for any sequence $(X^n)_{n\ge 1}$ of  $L^0(\R,\cF_T)$ that converges to $X$ in $\cO_t$, there exists a subsequence $(X^{k_n})_{n\ge 1}$   such that $X^{k_n}\ge \mu_t$ for some $\mu_t\in L^0(\R,\cF_t)$.
\end{definition}

\section{The NFL  and the NFLVR conditions}\label{NFL-VR}
Let us define $\cA_{t,T}:=\cV_{t,T}-L^0(\Reel_+,\cF_T)$ (resp. $\cA_{t,T}^c:=\cV_{t,T}^c-L^0(\Reel_+,\cF_T)$)  the set of all attainable claims from discrete-time (resp. continuous-time) portfolio processes. We denote by $L^{\infty}(\R,\cF_T)$ the set of all equivalence classes of bounded random variables $X$ such that $\| X\|_{\infty}<\infty$.  Consider the corresponding sets  $\cA_{t,T}^{\infty}:= \cA_{t,T}\cap L^{\infty}(\R,\cF_T)$ and $\cA_{t,T}^{c,\infty}:= \cA_{t,T}^c\cap L^{\infty}(\R,\cF_T)$  of  bounded attainable claims. Then, we denote by $\overline{\cA}_{t,T}^{w,\infty}$ and  $\overline{\cA}_{t,T}^{c,w,\infty}$ the  weak closures of $\cA_{t,T}^{\infty}$ and $\cA_{t,T}^{c,\infty}$ respectively with respect to the topology $\sigma(L^{\infty},L^1)$.   
\subsection{The NFL condition}

The NFL condition is very well known in mathematical finance. It means that it is not possible to asymptotically  get (in limit) a strictly positive profit when starting from a zero initial capital and following a bounded self-financing portfolio process. Here, asymptotically means that 
we complete the set of bounded self-financing portfolio processes by their limits in $L^\infty$ w.r.t. $\sigma(L^{\infty},L^1)$.

\begin{definition} Let $(\cO_t)_{t\le T}$ be a collection of topologies on $L^0(\R,\cF_T)$ and $\cV_{t,T}^c=\cV_{t,T}^c(\cO_t)$, $t\le T$.  The No Free Lunch condition (NFL, \cite{Kreps}) is defined at time $t$ by   $\overline{\cA}_{t,T}^{w,\infty}\cap L^{\infty}(\Reel^+,\cF_t)=\{0\}$ (resp. $\overline{\cA}_{t,T}^{c,w,\infty}\cap L^{\infty}(\Reel^+,\cF_t)=\{0\}$) for the model defined by the  discrete-time (resp. continuous-time) portfolio processes. We say that the NFL condition holds if it holds at any time $t\le T$.
\end{definition}

In the following, if $\cO$ and $\cO'$ are two topologies, we say that $\cO\subseteq\cO'$ if any open set of $\cO$ is an open set of $\cO'$. We  consider   a collection $(\cO_t)_{t\le T}$ of topologies on $L^0(\R,\cF_T)$ so that $\cV_{t,T}^c=\cV_{t,T}^c(\cO_t)$, $t\le T$.
\begin{lemm} \label{NFL-NFL0} Suppose that $\cO_0\subseteq\cO_t$ and $\cV_{t,T}\subseteq \cV_{0,T}$  for all $t\in [0,T]$. Then, the NFL condition holds for the continuous-time (resp.  discrete-time) portfolio processes  if and only if NFL holds at time $t=0$.
\end{lemm}
\begin{proof} By the assumptions, we deduce that $\cV_{t,T}^c\subseteq \cV_{0,T}^c$ for all $t\in [0,T]$. We deduce that $\overline{\cA}_{t,T}^{c,w,\infty}\subseteq \overline{\cA}_{0,T}^{c,w,\infty}$  for all $t\in [0,T]$.  The conclusion follows. \end{proof}

\begin{prop} \label{propNFLt} Suppose that the topology $\cO_t$, $t\le T$, satisfies the Fatou property, is $\cF_t$-positively homogeneous and is $\cF_t$-lower bond preserving. Assume that $\cV_{t,T}$ is a  $\cF_t$ positive cone, i.e. $\cV_{t,T}$ is convex and $\alpha_t\cV_{t,T}\subseteq \cV_{t,T}$ for all $\alpha_t\in L^0(\Reel^+,\cF_t)$.  Then, with $\cV_{t,T}^c=\cV_{t,T}^c(\cO_t)$,  the following statement are equivalent:

\begin{itemize}
\item [1.)] NFL holds at time $t$ for the  model defined by the discrete-time portfolio processes.

\item [2.)] There exists $Q_t\sim P$ such that $E_{Q_t}(V)\le 0$ for all $V\in \cV_{t,T}$ such that $V$ is bounded from below by a constant.

\item [3.)] NFL holds at time $t$ for the  model defined by the continuous-time portfolio processes.

\item [4.)] There exists $Q_t\sim P$ such that $E_{Q_t}(V)\le 0$ for all $V\in \cV_{t,T}^c$ such that $V$ is bounded from below by a constant.
\end{itemize}
\end{prop}
\begin{proof} By the assumptions, $\overline{\cA}_{t,T}^{w,\infty}$ and $\overline{\cA}_{t,T}^{c,w,\infty}$ are positive cones. Therefore the equivalences between 1.) and 2.) and between 3.) and 4.) are immediate consequences of the Kreps-Yan theorem, see \cite[Theorem 2.1.4]{KS}. Indeed, if $E_{Q_t}(V)\le 0$ for all $V\in \cA_{t,T}^c\cap L^{\infty}(\R,\cF_T)$ (resp.  $\cA_{t,T}\cap L^{\infty}(\R,\cF_T)$, it suffices to apply the Fatou lemma  to the sequence $V^m=V1_{\{V\le m\}}\in L^{\infty}(\R,\cF_T)$, as $m\to \infty$, if   $V$ is bounded from below, to deduce  4.) (resp. 2.)). It is clear that 4.) implies 2.) since $\cV_{t,T}\subseteq \cV_{t,T}^c$.  It remains to show that 2.) implies 4.). We first observe that 2.) implies that $E_{Q_t}(V \vert \cF_t)\le 0$ for all bounded from below $V\in \cV_{t,T}$, since $\cV_{t,T}$ is a $\cF_t$ positive cone. We then deduce by rescaling that the inequality $E_{Q_t}(V \vert \cF_t)\le 0$  also holds if $V$ is bounded from below by an $\cF_t$-measurable random variable. Then, consider $V\in \cV_{t,T}^c$ such that $V\ge m$ a.s. for some $m\in \Reel$. By definition,  $V=\lim_n V^n$ in $\cO_t$, for some convergent sequence of elements $V^n\in \cV_{t,T}$. As $\cO_t$ satisfies the Fatou property, we may suppose w.l.o.g. that $V\le \liminf_n V^n$. Moreover, as $\cO_t$ is $\cF_t$- lower bond preserving, we may also suppose that $V^n\ge \mu_t$ a.s., for all $n\ge 1$, where $\mu_t\in L^0(\Reel,\cF_t)$. Then,  $E_{Q_t}(V \vert \cF_t)\le \lim_n E_{Q_t}(V^n\vert\cF_t)$ by the Fatou lemma. As $E_{Q_t}(V^n\vert\cF_t)\le 0$ by the remark above, the conclusion follows. \end{proof}

\begin{remark} The equivalent probability measure $Q_t\sim P$ in Statemement 2.) is generally interpreted as a risk-neutral probability measure, see \cite{DMW}. 
\end{remark}

\begin{definition}\label{LB}The price process is said locally bounded if there exists a sequence of increasing stopping times $(T^n)_{n\ge 1}$ and a  real-valued sequence $(M^n)_{n\ge 1}$ such that $\lim_{n\to \infty}T^n=+\infty$ and  the stopped processes $S^{T^n}$ are bounded by $M^n$. 
\end{definition}

Note that, if the jumps $\Delta S_t=S_t-S_{t-}$ are uniformly bounded by a constant $M\ge 0$, it suffices to consider $T^n=\inf\{t\ge T_{n-1}:~S_t\ge n\}$ so that  $S^{T^n}\le M+n$. 

\begin{coro} \label{NFL-LocalMM} Suppose that $\cO_0\subseteq\cO_t$ for all $t\le T$. Suppose that the topology $\cO_0$ satisfies the Fatou property, is $\cF_0$-positively homogeneous and is $\cF_0$- lower bond preserving. Assume that $\cV_{t,T}$ is given by (\ref{Elmt2}) for all $t\le T$ and $S$ is a locally bounded process. Then, if  NFL holds for the discrete-time (resp. continuous-time) portfolios,  there exists a local martingale measure for $S$. Moreover, if $\cV_{t,T}=^{a}\!\!\cV_{t,T}$, the existence of a local martingale measure for $S$ implies NFL for both discrete-time and  continuous-time portfolios .
\end{coro}
\begin{proof}
Note that $\cV_{t,T}\subseteq \cV_{0,T}$ by (\ref{Elmt2}). Therefore, as $\cO_0\subseteq\cO_t$, it suffices to consider the NFL condition at time $t=0$ by Lemma \ref{NFL-NFL0}. By Proposition \ref{propNFLt}, NFL in discrete time and in continuous time are equivalent. In the following, we use the notations of Definition \ref{LB}. If NFL holds, the local martingale measure $Q=Q_0$ for $S$ is given 
by Proposition \ref{propNFLt}. Indeed, for each $n\ge 1$, and $t_1\le t_2$ such that $t_2\le T$, $V=\pm\left(S_{t_2\wedge T^n}-S_{t_1\wedge T^n}\right)1_{F_{t_1}}\in \cV_{0,T}$ for all $F_{t_1}\in \cF_{t_1}$ and $V$ is bounded from below by $-M^n$. So, we deduce that $E_Q(\left(S_{t_2\wedge T^n}-S_{t_1\wedge T^n}\right)1_{F_{t_1}})=0$ and finally $E_Q(S^{T^n}_{t_2}\vert \cF_{t_1})=S^{T^n}_{t_1}$. This implies that $S$ is a local martingale under $Q$. At last, if $\cV_{0,T}=^{a}\!\!\cV_{0,T}$, consider an admissible simple strategy $\theta$ such that $\cI_u(\theta)\ge m$ for all $u\in[0,T]$, see (\ref{IntegOperator}). Suppose that there exists a local martingale measure $Q$ for $S$. So, there exists an increasing sequence $(T^n)_{n\ge 1}$ of stopping times such that $\lim_n T^n=\infty$ and the stopped process $S^{T^n}$ is a martingale, for all $n\ge 1$. It is easily seen that $E_{Q}[\cI_{T\wedge T^n}(\theta)]= 0$. Indeed, it suffices to successively apply to tower property knowing that the generalized conditional expectation $E_{Q}(\theta_{\tau_{i-1}}\left(S_{\tau_{i}\wedge T^n}-S_{\tau_{i-1}\wedge T^n}\right)\vert\cF_{\tau_{i-1}})=0$. Moreover,  $\cI_{T\wedge T^n}(\theta)\ge m$ by the admissibility property. Therefore, $E_{Q}[\cI_{T}(\theta)]\le \liminf_n E_{Q}[\cI_{T\wedge T^n}(\theta)]\le 0$, by the Fatou lemma. The conclusion follows by  Proposition \ref{propNFLt}. \end{proof}

\subsection{The NFLVR condition}

The  NFLVR condition is also well known in mathematical finance. The financial interpretation is the same as the NFL one, i.e. it is an asymptotic no-arbitrage condition, but the topology $\sigma(L^{\infty},L^1)$ is replaced by the strong topology defined by the $L^{\infty}$ norm. \smallskip

 Let $\cA_{t,T}^{\infty}:= \cA_{t,T}\cap L^{\infty}(\R,\cF_T)$ and $\cA_{t,T}^{c,\infty}:= \cA_{t,T}^c\cap L^{\infty}(\R,\cF_T)$ be the  sets of  bounded attainable claims. Then, we denote by $\overline{\cA}_{t,T}^{\infty}$ and  $\overline{\cA}_{t,T}^{c,\infty}$ the  norm closures of $\cA_{t,T}^{\infty}$ and $\cA_{t,T}^{c,\infty}$ respectively with respect to the topology induced by the  norm $\| \cdot\|_{\infty}$.

\begin{definition} The condition NFLVR holds at time $t\le T$ for the discrete-time portfolios (resp. continuous-time portfolios) if $\overline{\cA}_{t,T}^{\infty}\cap L^{\infty}(\Reel_+,\cF_T)=\{0\}$ (resp. $\overline{\cA}_{t,T}^{c,\infty}\cap L^{\infty}(\Reel_+,\cF_T)=\{0\}$). We say that NFLVR holds if NFLVR holds at any time $t\le T$.
\end{definition}

We easily observe that NFL implies NFLVR. Actually, under some conditions on the price process, NFL and NFLVR are equivalent \cite[Corollary 1.2]{DScha1} to the existence of a local martingale measure, as we shall see. Note that it is not trivial whether the NFLVR condition for discrete-time portfolios is equivalent to the NFLVR condition for continuous-time portfolios. This is not true in general, see \cite[Example 6.5.]{DScha1}.  But we have the following:

\begin{prop} Suppose that $\cO_0\subseteq\cO_t$ for all $t\le T$. Suppose that the topology $\cO_0$ satisfies the Fatou property, is positively homogeneous and is $\cF_0$-lower bond preserving. Assume that $\cV_{t,T}=^a \!\! \cV_{t,T}$ is given by (\ref{Elmt2}) for all $t\le T$ and $S$ is a continuous process. Then,  the conditions  NFL and NFLVR for discrete-time portfolios and  the conditions  NFL and NFLVR for continuous-time portfolios  are equivalent to the existence of a local martingale measure for $S$. 
\end{prop}
\begin{proof} Recall that the NFL condition for discrete-time portfolios implies the NFLVR condition for discrete-time portfolios. By \cite[Theorem 7.6]{DScha1}, there exists a local martingale measure for $S$. By Corollary \ref{NFL-LocalMM}, we deduce that NFL holds both for discrete-time and continuous-time portfolio processes. The conclusion follows.
\end{proof}

The result above implies that the price process $S$ needs to be a semi-martingale for the NFL condition to hold. The same holds if the NFLVR condition holds even for discrete-time portfolio processes, see \cite[Theorem 7.2]{DScha1} for locally bounded processes $S$. The next no-arbitrage condition AIP we consider does not necessitate the price process to be a semimartingale. 

\section{The AIP condition}\label{sec4}

The AIP condition has been initially introduced in \cite{CL} for discrete-time models. The financial interpretation  is that the hedging prices of non negative European claims are non negative or, equivalently,  the hedging prices of non negative hedgeable European claims are finite. The advantage of this condition is that it is sufficient, at least in discrete-time, to deduce the super-hedging prices without supposing that the price process is a semimartingale. \smallskip

Our goal is to study the AIP condition for continuous-time processes and relate it to the same condition for discrete-time processes. To do so, we shall use the notion of conditional essential infimum  and supremum, see \cite[Section 5.3.1]{KS}. We recall that, if $\cH$ is a sub $\sigma$-algebra, the $\cH$-measurable essential supremum $\esssup_{\cH}(\Gamma)$ of a collection $\Gamma$ of real-valued random variables is the smallest $\cH$-measurable random variable that dominates $\Gamma$ a.s. and we define $\essinf_{\cH}(\Gamma)=-\esssup_{\cH}(-\Gamma)$. If the elements of $\Gamma$ are $\cH$-measurable, we use the notation $\esssup(\Gamma):=\esssup_{\cH}(\Gamma)$. If $\Gamma=\{\gamma\}$ is a singleton, we write $\esssup_{\cH} \Gamma=\esssup_{\cH} \gamma$.

\begin{theo} Let $\Gamma$ be a family of $\cF_T$-measurable random variables in $L^0(\R,\cF_T)$ and let $\cH$ be a sub $\sigma$-algebra of $\cF_T$. There exists a unique $\cH$-measurable random variable denoted by $\esssup_{\cH} \Gamma$ such that:
\begin{itemize}
\item [1)]  $\esssup_{\cH} \Gamma\ge \gamma$ a.s. for all $\gamma \in \Gamma$. \smallskip

\item [2)] If $\gamma_{\cH}$ is $\cH$-measurable and satisfies $\gamma_{\cH} \ge \gamma$ a.s. for all $\gamma \in \Gamma$, then $\gamma_{\cH} \ge \esssup_{\cH} \Gamma$ a.s..

\end{itemize}

\end{theo}
\begin{definition} \label{Price}

A contingent claim $h_{T}\in L^0(\R,\cF_T)$ is said to be super-hedgeable in discrete time (resp. continuous time) at time $t$ if there exists $p_{t}\in L^{0}(\R,\cF_{t})$ (called a super-hedging price) and a discrete-time (resp. continuous-time)  portfolio process  $ V_{t,T}$ such that $p_{t}+  V_{t,T} \ge h_T$. 
\end{definition}

Recall that the set of all  super-hedgeable claims in discrete time (resp. continuous time)  from the zero initial endowment at time $t$  is given by the set $\cA_{t,T}=\cV_{t,T}-L^0(\Reel_+,\cF_{T})$ (resp. $\cA_{t,T}^c$). We denote by $\cP_{t,T} (h_T)$ (resp. $\cP_{t,T} ^c(h_T)$)  the  set of super-hedging prices in discrete time (resp. in continuous time) for the claim $h_{T}\in L^0(\R,\cF_T)$. The infimum super-hedging price in discrete time (resp. in continuous time) is $\pi_{t,T} (h_T)= \essinf(\cP_{t,T} (h_T))$ (resp. $\pi_{t,T}^c (h_T)= \essinf(\cP_{t,T}^c (h_T))$).  We adopt the notation $\cP_{t,T} (0)=\cP_{t,T}$ (resp. $\cP_{t,T} ^c(0)=\cP_{t,T} ^c$), etc..when $h_T=0$. We observe that $\cP_{t,T}=  \cA_{t,T}\cap L^0( \Reel,\cF_{t})$ and $\cP_{t,T}^c= \cA_{t,T}^c\cap L^0( \Reel,\cF_{t})$. Moreover, 
  \bea \label{the prices} \cP_{t ,T}=\{\esssup_{\cF_{t}} ( -  v_{t,T}) :  v_{t,T} \in  \cV_{t,T} \}+ L^0( \Reel_+,\cF_{t}),\\
  \cP_{t ,T}^c=\{\esssup_{\cF_{t}} ( -  v_{t,T}) :  v_{t,T} \in  \cV_{t,T}^c \}+ L^0( \mathbb{\Reel_+},\cF_{t})
  .\eea
   Indeed, $p_t$ is a price in discrete time for $0$ if there exists $v_{t,T} \in  \cV_{t,T} $ such that $ p_t+  v_{t,T} \geq 0$ i.e $p_t\geqslant -  v_{t,T}$, which is equivalent to $ p_t \geqslant  \esssup_{\cF_{t}} (-  v_{t,T})$. We have a similar characterization for $\cP_{t ,T}^c$.

\begin{definition}
An instantaneous profit in discrete time (resp. in continuous time) at time $t<T$ is a strategy that super-replicates in discrete time (resp. in continuous time) the zero contingent claim  starting from a negative price $p_{t,T}\in  \cP_{t,T}\cap L^0(\Reel_-,\cF_t)$ (resp. $p_{t,T}\in  \cP_{t,T}^c\cap L^0(\Reel_-,\cF_t)$)   such that $p_{t,T}\ne 0$.  In the absence of such an instantaneous profit,  we say that the Absence of Instantaneous Profit (AIP)  holds at time $t$, i.e. \begin{equation}\label{NGD}
 \cP_{t,T}\cap L^{0}(\Reel_{-},\cF_t)=   \cA_{t,T}\cap L^0( \Reel_+,\cF_{t})=\{0\}.
\end{equation}
Respectively, $\cP_{t,T}^c\cap L^{0}(\Reel_{-},\cF_t)=   \cA_{t,T}^c\cap L^0( \Reel_+,\cF_{t})=\{0\}$ in continuous time. We say that AIP holds if AIP holds at any $t \le T$.
\end{definition}

\begin{remark}{\rm 
The NFLVR condition implies AIP.} $\Diamond$
\end{remark}

\begin{remark}{\rm 
 AIP in discrete time at time $t\le T$ is equivalent to $\pi_{t,T}(0) = 0$ or equivalently $\cP_{t,T} = L^{0}(\Reel_{+},\cF_t)$.
Indeed $\pi_{t,T}(0) \le 0$ as $0 \in  \cP_{t,T}$. Moreover, if AIP holds then $\cP_{t,T} \subset L^{0}(\Reel_+,\cF_t)$. To see it, consider  $p_{t,T} \in  \cP_{t,T}$. Then $1_{\{p_{t,T} \le 0\}} p_{t,T} \in  \cP_{t,T}$ hence $1_{\{p_{t,T} \le 0\}}p_{t,T}=0$ by AIP and $p_{t,T} \ge 0$. Conversely, any $p_t \ge 0$  is a price for the zero claim  since $0\in\cP_{t,T}$. The same holds in continuous time.} $\Diamond$
\end{remark}

The following lemma provides another financial interpretation of the AIP condition. Precisely, when starting from the zero initial endowment, it is not possible to obtain a  terminal wealth which, estimated at time $t$, is strictly positive on a non null $\cF_t$-measurable set. In particular, under AIP, there is a possibility to face a loss when starting from zero.

 \begin{lemm} \label{LemmAIP}The AIP condition holds in discrete time (resp. in continuous time) if and only if
 , for any $ t\le T$ and for all $ v_{t,T}\in \cV_{t,T}$ (resp. $\cV_{t,T}^c$), we have     $\essinf_{F_{t}} ( v_{t,T})\leqslant 0$.
\end{lemm}
\begin{proof}
This is a direct consequence of (\ref{the prices}). \end{proof}

\subsection{The AIP condition for discrete-time portfolio processes}

The following two propostions are direct consequences deduced from \cite{CL}.

\begin{prop} \label{AIP-St}
 Suppose that $d=1$ and the discrete-time portfolio processes are given by (\ref{Elmt1}). The  AIP condition holds in discrete time if and only if,   for all $ t_1< t_2<T$, $S_{t_1}\in\left[\essinf_{\cF_{t_1}}(S_{t_2}) , \esssup_{\cF_{t_1}}(S_{t_2})\right]$.\end{prop}

In the following, if $\cH$ is a sub $\sigma$-algebra, we denote by ${\rm supp}_{\cH}(X)$ the $\cH$-measurable conditional support of any random variable $X$, i.e. the smallest $\cH$-measurable random set ${\rm supp}_{\cH}(X)$ such that $X\in {\rm supp}_{\cH}(X)$ a.s., see \cite{EL}. The convex envelop of any $A\subseteq \Reel^d$ is denoted by ${\rm conv}(A)$.

\begin{prop}\label{prop-AIP-dt} Suppose that $d \ge 1$ and  the  discrete-time portfolio processes are given by (\ref{Elmt1}).
Then, AIP holds in  discrete time if and only if $S_{t_1} \in {\rm conv}({\rm supp}_{\cF_{t_1}}(S_{t_2}))$ for any $t_1 \le t_2 \le T$.
\end{prop}

Similarly, we may show the following:

\begin{prop}\label{Prop-AIP-SP} Suppose that the discrete-time portfolio processes are given by (\ref{Elmt2}). Then, AIP holds in  discrete time if and only if $S_{\tau_1} \in {\rm conv}({\rm supp}_{\cF_{\tau_1}}(S_{\tau_2}))$ for every stopping times $\tau_1,\tau_2\in \cT_{0,T}$ such that $\tau_1 \le \tau_2 $.
\end{prop}
\begin{proof}
Suppose that AIP holds and consider two stopping times $\tau_1 \le \tau_2$ in $[0,T]$.  Then, AIP holds for the two time steps smaller model  defined by $(S_{\tau_i})_{i=1,2}$ and $(\cF_{\tau_i})_{i=1,2}$. By \cite{CL}, we deduce that the minimal price of the zero claim for $(S_{\tau_i})_{i=1,2}$ is given by

$$ 0=\pi_{\tau_1,\tau_2}(S_{\tau_1},S_{\tau_2})=- \delta_{ {\rm conv}({\rm supp}_{\cF_{\tau_1}}(S_{\tau_2}))}(S_{\tau_1}),$$
where, for any $I\subseteq \Reel^d$, $\delta_I=(+\infty
)1_{I}$ with the convention $(+\infty)\times (0)=0$. Therefore, $S_{\tau_1} \in {\rm conv}({\rm supp}_{\cF_{\tau_1}}(S_{\tau_2})$. 

Reciprocally,  suppose that, for any $\tau_1 \le \tau_2 \le T$, $S_{\tau_1} \in {\rm conv}({\rm supp}_{\cF_{\tau_1}}(S_{\tau_2})$. Then,  $0=\pi_{\tau_1,\tau_2}(S_{\tau_1},S_{\tau_2})$ for any $\tau_1 \le \tau_2 \le T$. Consider $ p_t \in \cP_{t,T}$ such that $p_{t}+\sum_{i=1}^n \theta_{\tau_{i-1}} \Delta S_{\tau_i} \ge 0 $  for some strategies   $\theta_{\tau_i}\in L^0( \mathbb{R}^d,\cF_{\tau_i})$  and stopping times $t=\tau_0<\tau_1<\cdots<\tau_n=T$. Then, $p_{t}+ \sum \limits_{i=0}^{n-2} \theta_{\tau_{i}} \Delta S_{\tau_{i+1}}$ is a price for the zero claim in the  two time steps model $(S_{\tau_i})_{i=n-1,n}$. As $0=\pi_{\tau_{n-1},\tau_n}(S_{\tau_{n-1}},S_{\tau_n})$, we get that $p_{t}+ \sum \limits_{i=0}^{n-2} \theta_{\tau_{i}} \Delta S_{\tau_{i+1}}\ge 0$. By induction, we finally deduce that $p_t\ge 0$, i.e. AIP holds.
\end{proof}\smallskip

We know reformulate the proposition above when $d=1$ in term of sub-maxingales, see \cite{BCJ}.

\begin{definition}We say that a continuous-time process $M=(M_{t})_{t\le T}$ adapted to the filtration $(\cF_t)_{t\in [0,T]}$ is a sub-maxingale (resp. super-maxingale)  if, for any $u,t\in [0,T]$ such that $u\le t$, we have $\esssup_{\cF_u}M_t\ge M_u$ (resp. we have $\esssup_{\cF_u}M_t\le M_u$). Moreover, $M$ is said a maxingale if it is both a super-maxingale and a sub-maxingale.
\end{definition}

Note that the notion of maxingale is an adaptation of the martingale concept to the conditional supremum operator. Observe that, for a super-maxingale $M$, $\esssup_{\cF_u}M_t\le M_u$ implies that $M_u\ge M_t$ and we deduce that the  super-maxingales coincide with the non  increasing processes.

\begin{definition}We say that a continuous-time process $M=(M_{t})_{t\le T}$ adapted to the filtration $(\cF_t)_{t\in [0,T]}$ is a strong sub-maxingale if, for any $\tau\in \cT_{0,T}$, the stopped process $M^{\tau}$ is a sub-maxingale.
\end{definition}

An open issue is whether a  sub-maxingale may be a strong sub-maxingale. When the operator is the conditional expectation, the Doob's stopping Theorem \cite{JS} states that this is the case, at least when $M$ is bounded from above by a martingale, see \cite[Theorem 1.39]{JS}. By Lemma \ref{maxingale5}, we have:

\begin{prop}\label{S-SM} Let $M=(M_{t})_{t\le T}$  be a  right-continuous continuous-time process adapted to the filtration $(\cF_t)_{t\in [0,T]}$. Then, $M$ is a strong sub-maxingale if and only if for all stopping times $ \tau,S \in \cT_{0,T}$,  $ \esssup_{\cF_{S}}(M_{\tau }) \ge M_{ S\wedge \tau }$.\end{prop} 
\begin{proof}
Suppose that $M$ is a strong sub-maxingale. Let $S,\tau \in \cT_{0,T}$. As $S^{\tau}$ is a sub-maxingale, we apply Lemma \ref{maxingale5} with the stopping time $S$ and the deterministic stopping time $T$. We get that $ \esssup_{\cF_{S}}(M_{\tau\wedge T })\ge M_{\tau \wedge S\wedge T}$, i.e. 
 $ \esssup_{\cF_{S}}(M_{\tau })\ge M_{\tau \wedge S}$. The reverse implication is immediate. \end{proof}

\begin{prop} Suppose that $d=1$ and the discrete-time portfolio processes are given by (\ref{Elmt2}). The following statements are equivalent:
\begin{itemize}  
\item [1.)]AIP condition holds in discrete-time.\smallskip

\item [2.)] We have $S_{\tau_1}\in\left[\essinf_{\cF_{\tau_1}}(S_{\tau_2}) , \esssup_{\cF_{\tau_1}}(S_{\tau_2})\right]$, for all $\tau_1, \tau_2 \in \cT_{0,T}$ such that  $ \tau_1\le  \tau_2$. \smallskip

\item [3.)] $S$ and $-S$ are strong sub-maxingales. 
\end{itemize}

\end{prop}
\begin{proof}
Suppose that AIP holds. Condition AIP for the    discrete-time portfolios of  (\ref{Elmt2}) implies the statement 2.) by Proposition \ref{Prop-AIP-SP}. In particular,  $S$ and $-S$ are  sub-maxingales and, for any $t\in [0,T]$ and  $\tau \in \cT_{0,T}$ such that $\tau\ge t$ a.s., we have by 2.) the inequality
\bea \label{IneqSubMax} \esssup_{\cF_t}S_{\tau}\ge S_{t}.\eea For fixed $\tau \in \cT_{0,T}$, we deduce that $S^{\tau}$ is a sub-maxingale. To see it, consider $t_1<t_2\le T$. On the set $A=\{\tau\wedge t_2<t_1\}\in \cF_{t_1}$, we have 
$$1_{A}\esssup_{\cF_{t_1}}S_{t_2}^{\tau}=1_{A}\esssup_{\cF_{t_1}}S_{t_2\wedge\tau\wedge t_1}=1_{A}S_{\tau\wedge t_1}.$$
On $B=\Omega\setminus A$, as $(t_2\wedge\tau)\vee t_1\ge t_1$, we deduce from (\ref{IneqSubMax})  that 
$$1_{B}\esssup_{\cF_{t_1}}S_{t_2}^{\tau}=1_{B}\esssup_{\cF_{t_1}}S_{(t_2\wedge\tau)\vee t_1}\ge 1_{B}S_{t_1}=1_{B}S_{t_1\wedge\tau}.$$
Therefore, we conclude that $\esssup_{\cF_{t_1}}S_{t_2}^{\tau}\ge S_{t_1}^{\tau}$ and, finally, $S$ is a strong sub-maxingale.  By the same reasoning,  $-S$ is also a strong sub-maxingale. Therefore, 1.) implies 2.), which implies 3.). Moreover, 3.) implies 2.) by  Proposition \ref{S-SM}. At last, if 2.) holds, we conclude that 1.) holds by Proposition \ref{Prop-AIP-SP}.
\end{proof}

\subsection{The AIP condition for continuous-time portfolio processes}\label{AIP-continuous}

In  this section, we consider  topologies $(\cO_t)_{t\in [0,T]}$ such that $\cV_{t,T}^c=\cV_{t,T}^c(\cO_t)$ for all $t\le T$, and such that the AIP condition in continuous time and  in discrete time are equivalent, as stated in our main Theorem \ref{MR}. Precisely, we consider for any time $t\le T$, the topology on $L^0(\R,\cF_T)$ induced by the pseudo-distance:

\bea \label{PDIST} \hat d^+_t (X,Y)=E(\esssup_{\cF_t}(X-Y)^+\wedge  1  ),\quad X,Y \in  L^0( \Reel,\cF_{T}).\eea We send the readers to  Section \ref{PseudoDistance} for the  definition and the main properties of a pseudo-distance topology.\smallskip

We notice that a sequence of discrete-time portfolios $(V_{t,T}^n)_{n\ge 1}$ of $\cV_{t,T}$ is convergent in $\cO_t$ if and only if $\inf_{n\ge 1}V_{t,T}^n>-\infty$ a.s., see Proposition \ref{CriteriaConv}. So, $\cV_{t,T}^c=\cV_{t,T}^c(\cO_t)$ is an a priori large class of so-called continuous-time portfolios. In particular, if $(V_{t,T}^n)_{n\ge 1}$ is a  sequence of usual stochastic integrals that converge to some stochastic integral $\cI_{t,T}(\theta)$, then the convergence holds in probability hence so does in $\cO_t$ by Proposition \ref{CriteriaConv}. Any limit $V_{t,T}^c\in \cV_{t,T}^c$ satisfies $V_{t,T}^c\le \cI_{t,T}(\theta)$ by Proposition \ref{liminf} but $\cI_{t,T}$ does not necessarily belong to $V_{t,T}^c$. This means that  $\cI_{t,T}$ cannot necessarily be super-hedged asymptotically by simple strategies.

 Let us give a financial interpretation of the convergence in $\cO_t$. By Proposition \ref{set-of-limits}, $V_{t,T}^n$ converges to $V_{t,T}^c\in \cV_{t,T}^c$ if $V_{t,T}^c\le V_{t,T}^n +\alpha_t^n$ for all $n\ge 1$, where $\alpha_t^n\in L^0(\Reel_+,\cF_t)$ converges to $0$ in probability. Therefore, it is possible to reach (actually super-replicates) the continuous-time portfolio value $V_{t,T}^c$ from discrete-time portfolios up to an arbitrary small error. This is why we believe that this topology is well adapted to finance. By Proposition \ref{liminf}, Proposition \ref{linear+} and Proposition \ref{set-of-limits}, we obtain that $\cO_t$ satisfies the Fatou property, is $\cF_t$-positively homogeneous and is  $\cF_t$-low bound preserving. This implies that the NFL and the NFLVR conditions in discrete-time and continuous-time are equivalent as stated in Section \ref{NFL-VR} for these pseudo-distance topologies. We 
also have:

\begin{lemm}Suppose that, for any $t\le T$, $\cO_t$ is the pseudo-distance topology defined by (\ref{PDIST}) and $\cV_{t,T}^c=\cV_{t,T}^c(\cO_t)$. Then, the NFLVR condition in continuous-time is equivalent to the NA condition  $\cA_{t,T}^{c}\cap L^{0}(\Reel_+,\cF_T)=\{0\}$ in continuous-time, for all $t\le T$.
\end{lemm}
\begin{proof}
Notice that by Proposition \ref{set-of-limits}, $\cA_{t,T}^{c,\infty}$ is closed in $L^{\infty}$ hence we have  $\overline{\cA}_{t,T}^{c,\infty}=\cA_{t,T}^{c,\infty}$ and NFLVR reads as $\cA_{t,T}^{c}\cap L^{\infty}(\Reel_+,\cF_T)=\{0\}$, which is equivalent to the NA condition as $\cA_{t,T}^{c}-L^0(\Reel_+,\cF_T)\subseteq \cA_{t,T}^{c}$.
\end{proof}

The main result of this section is the following:

\begin{theo}\label{MR} Suppose that, for any $t\le T$, $\cO_t$ is the pseudo-distance topology defined by (\ref{PDIST}) and $\cV_{t,T}^c=\cV_{t,T}^c(\cO_t)$. Then, AIP holds in continuous time  if and only if AIP holds in discrete time.
\end{theo}
\begin{proof} It suffices to prove that AIP holds in continuous time if it holds in discrete time. By Lemma \ref{LemmAIP}, we have 
$\essinf_{F_{t}} ( v_{t,T})\leqslant 0$ for all $v_{t,T}\in \cV_{t,T}$. We have to show the same for $v_{t,T}^c\in \cV_{t,T}^c$. By Proposition \ref{set-of-limits}, $V_{t,T}^c\le V_{t,T}^n +\alpha_t^n$ for all $n\ge 1$, where $\alpha_t^n\in L^0(\Reel_+,\cF_t)$ converges to $0$ in probability and $V_{t,T}^n\in \cV_{t,T}$.  As $\alpha_t^n$ is $\cF_t$-measurable, we deduce that 
$$\essinf_{F_{t}} V_{t,T}^c\le \essinf_{F_{t}} V_{t,T}^n +\alpha_t^n\le \alpha_t^n.$$
As $n\to +\infty$, we deduce that $\essinf_{F_{t}} V_{t,T}^c\le 0$ hence AIP holds in continuous time by Lemma \ref{LemmAIP}.
\end{proof}

\section{The NUPBR no-arbitrage condition}

The No Unbounded Profit with Bounded Risk no-arbitrage condition NUPBR has been introduced in \cite{KK}. In our setting, this condition may be adapted if we only consider admissible portfolios. This is why, we suppose that the portfolio processes are generated by an operator $\cI$ as in Remark \ref{Operator}. We define for $m\in (0,\infty)$, $^{a}\cV_{t,T}(m)$ (resp. $^{a}\cV_{t,T}^c(m)$ in continuous time) the set of all admissible portfolio values $V_{t,T}=\cI(\theta)\in ^{a}\!\!\cV_{t,T}$ such that  $V_{t,u}=\cI_u(\theta)\ge -m$ for all $u\in[t,T]$.  

Let us define, for every $t\in [0,T]$, the space ${\rm SP}(\R,(\cF_u)_{u\in [t,T]})$   of all $(\cF_u)_{u\in [t,T]}$-adapted real-valued stochastic processes on $[t,T]$. We  consider the family of topologies $(\cO_t)_{t\in [0,T]}$ such that, for every $t\in [0,T]$, $\cO_t$ is the topology on  ${\rm SP}(\R,(\cF_u)_{u\in [t,T]})$ which is  induced by the pseudo-distance:
\bea   &&\hat d^+_t (X,Y)=E(\esssup_{u\in [t,T]}\esssup_{\cF_t}(X_u-Y_u)^+\wedge  1  ),\label{PDIST1}\\
&& X,Y \in  {\rm SP}(\R,(\cF_u)_{u\in [t,T]}).\nonumber\quad \quad\eea 

By the same reasoning as in the proof of Proposition \ref{set-of-limits}, a sequence $(X^n)_{n\ge 1}\in {\rm SP}(\R,(\cF_u)_{u\in [t,T]})$ converges to $X\in {\rm SP}(\R,(\cF_t)_{u\in [u,T]})$ in  $\cO_t$ if and only if there exists a sequence $(\alpha_t ^n)_{n\ge 1}$ such that $\alpha_t ^n$ tends to $0$ in probability as $n\to \infty$ and $X_u\le X^n_u+\alpha_t ^n$ for all $u\in [t,T]$. Moreover, adapting the Proposition \ref{CriteriaConv}, we may show that a sequence $(X^n)_{n\ge 1}\in {\rm SP}(\R,(\cF_u)_{u\in [t,T]})$ is convergent in $\cO_t$ if and only if $\inf_n X^n_u>-\infty$ a.s. for all $u\in [t,T]$. \smallskip

With $(\cO_t)_{t\in [0,T]}$ given by (\ref{PDIST1}), we define $\cV_{t,T}^c$ as the terminal values $V_{t,T}^c(T)$ of limit processes $V_{t,T}^c$ such that $V_{t,T}^c=\lim_n V_{t,T}^n$ 
where $V_{t,T}^n=(V_{t,T}^n(u))_{u\in [t,T]}$ are the discrete time processes associated to $\cV_{t,T}$, see Remark \ref{Operator}.

\begin{definition} We say that  NUPBR holds in discrete time (resp. in continuous time) at time $t\le T$ if, for any $m>0$, $^{a}\cV_{t,T}(m)$ (resp. $
^{a}\cV_{t,T}^c(m)$) is bounded in probability. We say that NUPBR holds if it holds at any time.
\end{definition}

 Recall that a sequence $(X^n)_{n\ge 0}$ of random variables is bounded in probability if, for all $\epsilon>0$, there exists $n_0\ge 1$ and $M>0$ such that, for all $n\ge n_0$, $P(\vert X^n\vert>M)\le \epsilon$ . More generally, a set $C\subseteq L^0(\R,\cF_T)$ is bounded in probability if any sequence $(X^n)_{n\ge 0}$ of $C$ is bounded in probability.\smallskip

In the setting of semimartingales, it is shown in \cite{KK} that NUPBR + NA, i.e. $\cV_{0,T}\cap L^0(\Reel^+,\cF_T)=\{0\}$, is equivalent to NFLVR. In particular, NUPBR alone does not necessarily implies NA. This is due to the fact that a portfolio $V_{t,T}\in \cV_{t,T}$ such that $V_{t,T}\ge 0$ is not necessary admissible. Otherwise, if $V_{t,T}$ is admissible, then by \cite[Theorem 3.12]{KK}, we get that $V_{t,T}(u)\ge 0$ for all $u\in [t,T]$ by the super-martingale property. Then, necessary $V_{t,T}=0$, i.e. NA would hold since, otherwise, the sequence $V_{t,T}^n=nV_{t,T}$, $n\ge 1$, is unbounded in probability.  In conclusion, NUPBR holds at time $t$ in continuous time (resp. in discrete time) implies NA (and so AIP) at time $t$ only for the restricted sets $^a\cV_{t,T}^c$ and $^a\cV_{t,T}$ respectively.\smallskip

 Our main result of this section is the following. Before, we recall a definition:
 
 \begin{definition} We say that a subset $\Gamma$ of $L^0(\R,\cF_T)$ is infinitely \\ $\cF_t$-decomposable (resp. $\cF_t$-decomposable) if for any partition of $\Omega$ (resp. finite partition) by elements $(F_t^n)_{n=1}^{\infty}$ of $\cF_t$ and any sequence $(X^n)_{n\ge 1}$ of $\Gamma$, we have $\sum_{n=1}^\infty X^n1_{F_t^n}\in \Gamma$. 
\end{definition}

\begin{theo}\label{SNUPBR} Suppose that, for  $t\le T$, $\cO_t$ is the pseudo-distance topology defined by (\ref{PDIST1}) and $\cV_{t,T}^c=\cV_{t,T}^c(\cO_t)$. Suppose that  $\cV_{t,T}$ is infinitely $\cF_t$-decomposable. Then,  NUPBR holds in discrete time if and only if it holds in continuous time.
\end{theo}
\begin{proof} It suffices to show that NUPBR holds in continuous time if it holds  in discrete time. To do so, suppose that $
^{a}\cV_{t,T}^c(m)$ is not bounded in probability for some $m>0$. Then, there exists a sequence $(V_{t,T}^{c,n})_{n\ge 1}\in 
^{a}\!\!\!\cV_{t,T}^c(m)$ and $\epsilon\in (0,1)$ such that $P(V_{t,T}^{c,n}>n)>\epsilon$ for all $n\ge 1$. By Proposition \ref{set-of-limits}, for all $n\ge 1$, there exists a sequence $(V_{t,T}^{n,m})_{m\ge 1}\in ^a\!\!\!\cV_{t,T}$ and a sequence $(\alpha^{n,m}_t)_{m\ge 1}\in L^0(\Reel_+,\cF_t)$ such that $\alpha^{n,m}_t$ converges to $0$ in probability as $m\to \infty$ and $V_{t,T}^{c,n}(u)\le V_{t,T}^{n,m}(u)+\alpha^{n,m}_t$, for all $m\ge 1$ and $u\in [0,T]$. We may assume w.l.o.g. that $\alpha^{n,m}_t$ converges to $0$ a.s. as $m\to \infty$. Then, there exists an integer-valued $\cF_t$-measurable random variable  $m^n_t$ such that $\alpha^{n,m^n_t}_t\in L^0([0,1],\cF_t)$. As $\cV_{t,T}$ is infinitely $\cF_t$-decomposable, we deduce that $V_{t,T}^{n,m^n_t}\in \cV_{t,T}$. Note that  $V_{t,T}^{c,n}(u)\le V_{t,T}^{n,m^n_t}(u)+1$ hence  $V_{t,T}^{n,m^n_t}\in  ^{a}\!\!\cV_{t,T}(m+1)$ for all $n\ge 1$. Moreover,  $\epsilon<P(V_{t,T}^{c,n}>n)\le P(V_{t,T}^{n,m^n_t}>n-1)$, for all $n\ge 1$. This implies that the sequence $(V_{t,T}^{n,m^n_t})_{n\ge 1}$ is not bounded in probability, contrarily to the assumption NUPBR for $\cV_{t,T}$. This contradiction allows one to conclude that NUPBR holds in continuous time.
\end{proof}

\section{Super-hedging prices}

\subsection{Super-hedging prices without no-arbitrage condition}

Recall that the super-hedging prices (resp. the infimum super-hedging price)  of a payoff $h_T\in L^0(\R,\cF_T)$ are defined after Definition \ref{Price}. Our main result is the following:

\begin{theo} \label{theoPrice1} Suppose that, for any $t\le T$, $\cO_t$ is the pseudo-distance topology defined by (\ref{PDIST}) and $\cV_{t,T}^c=\cV_{t,T}^c(\cO_t)$. Then, the infimum super-hedging prices of a payoff $h_T\in L^0(\R,\cF_T)$, in discrete time and in continuous time respectively,   coincide i.e. 
$$\pi_{t,T} (h_T)= \essinf(\cP_{t,T} (h_T))=\pi_{t,T}^c (h_T)= \essinf(\cP_{t,T}^c (h_T)).$$
\end{theo}
\begin{proof}
As $\cP_{t,T} (h_T)\subseteq \cP_{t,T}^c (h_T)$, we have $\pi_{t,T}^c (h_T)\le \pi_{t,T} (h_T)$.  Consider a price $p_t\in \cP_{t,T}^c (h_T)$ such that $p_t+V_{t,T}^c\ge h_T$ for some $V_{t,T}^c\in \cV_{t,T}^v$. By Proposition \ref{set-of-limits}, we have $V_{t,T}^c\le V_{t,T}^n +\alpha_t^n$ for all $n\ge 1$, where $\alpha_t^n\in L^0(\Reel_+,\cF_t)$ converges to $0$ in probability and $V_{t,T}^n\in \cV_{t,T}$. We deduce that $p_t+\alpha_t^n\in \cP_{t,T} (h_T)$ hence $p_t+\alpha_t^n\ge \pi_{t,T} (h_T)$. As $n\to \infty$, we deduce that 
$p_t\ge \pi_{t,T} (h_T)$ hence $\pi_{t,T}^c (h_T)\ge \pi_{t,T} (h_T)$. The conclusion follows.
\end{proof}

\begin{remark}{\rm \quad \smallskip

1.) Note that, at any time, $\cP_{t,T}^c (h_T)$ may be empty. In that case, we also have $\cP_{t,T}^c (h_T)=\emptyset$ and $\pi_{t,T} (h_T)=\pi_{t,T}^c (h_T)=\infty$. Reciprocally, if we have $\cP_{t,T}^c (h_T)=\emptyset$, then $\pi_{t,T} (h_T)=\infty$ and we deduce that $\pi_{t,T}^c (h_T)=\infty$ by Theorem \ref{theoPrice1}.\smallskip

2.) If $\cV_{t,T}$ is a positive cone, then $\cP_{t,T}$ and $\cP_{t,T}^c$ are positive cones if $\cO_t$ is $\cF_t$-positively homogeneous. Therefore, $\pi_{t,T}=\pi_{t,T}(0)<0$ implies that $\pi_{t,T}=\pi_{t,T}^c (h_T)=-\infty$. Let us consider a payoff $h_T\in L^0(\R,\cF_T)$ such that $h_T\le \alpha S_T + \beta$ for some $\alpha,\beta\in \Reel$. Then, for all price $p_{t,T}\in \cP_{t,T}$, we deduce that $p_{t,T}+\alpha S_t + \beta \in \cP_{t,T}(h_T)$. Therefore, $\pi_{t,T}=\pi_{t,T}^c=-\infty$ implies that $\pi_{t,T}(h_T)=\pi_{t,T}^c(h_T)=-\infty$. This is why the condition AIP is financially meaning as it avoids this unrealistic situation where the prices of a positive payoff $h_T$ may be as negatively large as possible so that it is not possible to compute the infimum price.\smallskip

3.) If $\cV_{t,T}=\cup_{n\ge n} \cV_{t,T}(n)$ where $\cV_{t,T}(n)$ is an increasing sequence of discrete-time models, then observe that $\pi_{t,T}(h_T)=\inf_n\pi_{t,T}^n(h_T)$ where $\pi_{t,T}^n(h_T)$ are the infimum prices associated to the models $\cV_{t,T}(n)$, $n\ge 1$. Moreover, if $\cV_{t,T}(n)$ is a model only composed of a finite number of dates, then $\pi_{t,T}^n(h_T)$ may be computed as in \cite{CL}. This is the case in practice, if the trades only may be executed at deterministic dates, e.g. every second.
}$\Diamond$\end{remark}

\subsection{Infinitely $\cF_t$-decomposable extension of the discrete-time prices}

In the following, we show that the discrete-time portfolio processes may be extended without changing the infimum prices and we get a precise form of the set of super-hedging prices. We denote by ${\rm Part}_t(\Omega)$ the set of all $\cF_t$-measurable partitions of   $\Omega$ and we consider

\bea
\cV_{t,T}^{{\rm id}}=\left\{\sum_{n=1}^\infty X^n1_{F_t^n}:~X^n\in \cV_{t,T},\, (F_t^n)_{n=1}^{\infty}\in  {\rm Part}_t(\Omega)  \right\}.
\eea

Note that $\cV_{t,T}^{{\rm id}}$ is  infinitely  $\cF_t$-decomposable. We say that $\cV_{t,T}^{{\rm id}}$ is the discrete-time infinitely  $\cF_t$-decomposable extension of $\cV_{t,T}$.  We then denote by $\cP_{t,T}^{{\rm id}}(h_T)$ the set of all prices obtained from $\cV_{t,T}^{{\rm id}}$ and $\pi_{t,T}^{{\rm id}}(h_T):=\essinf_{\cF_t}\cP_{t,T}^{{\rm id}}(h_T)$. We denote by $\cV_{t,T}^{{\rm id},c}$ the continuous-time processes deduced from $\cV_{t,T}^{{\rm id}}$.

\begin{lemm} The AIP condition holds for $\cV_{t,T}$ if and only AIP holds for its infinitely  $\cF_t$-decomposable extension.
\end{lemm}
\begin{proof}

It suffices to show that  AIP holds for its $\cF_t$-decomposable extension as soon as it holds for $\cV_{t,T}$. By Lemma \ref{LemmAIP}, let us show that $\essinf_{\cF_t}(V_{t,T})\le 0$ for all $V_{t,T}\in \cV_{t,T}^{{\rm id}}$. Suppose that $\cV_{t,T}^{{\rm id}}=\sum_{n=1}^\infty X^n1_{F_t^n}$ where $X^n\in \cV_{t,T}$ and $(F_t^n)_{n=1}^{\infty}\in  {\rm Part}_t(\Omega) $. Then, 
\bean 1_{F_t^m}\essinf_{\cF_t}(V_{t,T})=1_{F_t^m}\essinf_{\cF_t}(V_{t,T}1_{F_t^m})=1_{F_t^m}\essinf_{\cF_t}(X^n)\le 0.\eean
The conclusion follows.
\end{proof}

\begin{lemm} Suppose that $\cV_{t,T}$ is   $\cF_t$-decomposable, $t\le T$, and consider a payoff  $h_T\in L^0(\R,\cF_T)$. Then, we have $\pi_{t,T}^{{\rm id}}(h_T)=\pi_{t,T}(h_T)$ and 
$$\cP_{t,T}(h_T)\subseteq \cP_{t,T}^{{\rm id}}(h_T)\subseteq \overline{\cP}_{t,T}(h_T),$$
where $\overline{\cP}_{t,T}(h_T)$ is the closure of $\cP_{t,T}(h_T)$ in $L^0$.
\end{lemm}
\begin{proof}
As $\cV_{t,T}\subseteq \cV_{t,T}^{{\rm id}}$, we have $\cP_{t,T}(h_T)\subseteq \cP_{t,T}^{{\rm id}}((h_T)$ and $\pi_{t,T}^{{\rm id}}(h_T)\le \pi_{t,T}(h_T)$. Moreover, if $p_t\in \cP_{t,T}^{{\rm id}}((h_T)$, then we have $p_t+\sum_{i=1}^{\infty}V_{t,T}^i1_{F_t^i}\ge h_T$ for some $V_{t,T}^i\in \cV_{t,T}$, $i\ge 1$ and a partition $(F_t^i)_{i\ge 1}$ of $\Omega$ by elements of $\cF_t$. Consider $p_t^0\in \cP_{t,T}(h_T)$ and define $p_t^n=p_t1_{\cup_{i=1}^nF_t^i}+p_t^01_{\Omega\setminus \cup_{i=1}^nF_t^i}$, $n\ge 1$. As $\cV_{t,T}$ is   $\cF_t$-decomposable, $p_t^n\in \cP_{t,T}(h_T)$. Moreover, $p_t=\lim_{n\to \infty}p_t^n$. We then deduce that  $\cP_{t,T}^{{\rm id}}((h_T)\subseteq \overline{\cP}_{t,T}(h_T)$ hence $\pi_{t,T}^{{\rm id}}(h_T)\ge \pi_{t,T}(h_T)$. The conclusion follows. \end{proof}

\begin{prop} Consider a payoff  $h_T\in L^0(\R,\cF_T)$. Then, there exists $ \Lambda_t\in \cF_t$ such that  $\cP_{t,T}^{{\rm id}} (h_T)=L^0(J_{t,T}(h_T),\cF_t)$ and
$$J_{t,T}(h_T)=[\pi_{t,T}^{{\rm id}}  (h_T),\infty)1_{\Lambda_t}+(\pi_{t,T}^{{\rm id}}  (h_T),\infty)1_{\Omega\setminus\Lambda_t}.$$
\end{prop}
\begin{proof} It suffices to argue on the set of all $\omega$ such that  $\pi_{t,T}^{{\rm id}}  (h_T)<\infty$. Therefore, we suppose w.l.o.g. that there exists $p_t^0\in \cP_{t,T}^{{\rm id}} (h_T)$. Let us consider
$$\Gamma_t=\left\{ \Lambda_t\in \cF_t:~   \pi_{t,T}^{{\rm id}}  (h_T)1_{ \Lambda_t}+ p_t^0   1_{\Omega\setminus\Lambda_t}\in   \cP_{t,T}^{{\rm id}} (h_T)\right\}.$$
Note that $\emptyset\in \Gamma_t$. As $\cV_{t,T}^{{\rm id}}$ is   infinitely  $\cF_t$-decomposable,  $\cP_{t,T}^{{\rm id}} (h_T)$ is infinitely  $\cF_t$-decomposable   by Lemma \ref{decomp}. We deduce that $\Lambda_t^1\cup \Lambda_t^2 \in \Gamma_t$ if $\Lambda_t^1,\Lambda_t^2 \in \Gamma_t$. Then, the family $\{1_{\Lambda_t}:\Lambda_t\in \Gamma_t\}$ is directed upward. We deduce that  $\esssup_{\Lambda_t\in \Gamma_t}1_{\Lambda_t}=1_{\Lambda_t^{\infty}}$ where $\Lambda_t^{\infty}$ is an increasing union of elements of $\Gamma_t$. As $\cP_{t,T}^{{\rm id}} (h_T)$ is infinitely  $\cF_t$-decomposable, we get that $\Lambda_t^{\infty}\in \Gamma_t$. We may also show that $\Lambda_t^{\infty}$ is independent of $p_t^0$. We then define $J_{t,T}(h_T)$ as above with $\Lambda_t=\Lambda_t^{\infty}$. We claim that $\cP_{t,T}^{{\rm id}} (h_T)=L^0(J_{t,T}(h_T),\cF_t)$. To see it, consider a price $p_t^0\in \cP_{t,T}^{{\rm id}} (h_T)$ and suppose that $p_t^0=\pi_{t,T}^{{\rm id}}  (h_T)$ on a non null set of $\Omega\setminus\Lambda_t$. Then, we get a contradiction with the maximality of $\Lambda_t$. So, we obtain that $\cP_{t,T}^{{\rm id}} (h_T)\subseteq L^0(J_{t,T}(h_T),\cF_t)$. Reciprocally, consider $p_t\in L^0(J_{t,T}(h_T),\cF_t)$. Then, $p_t^0=p_t+1_{\Lambda_t}>\pi_{t,T}^{{\rm id}}  (h_T)$ a.s. hence $p_t^0\in \cP_{t,T}^{{\rm id}} (h_T)$ by  Lemma \ref{PriceCriteria}. Moreover, $p_t\ge  \pi_{t,T}^{{\rm id}}  (h_T)1_{ \Lambda_t}+p_t^01_{\Omega\setminus\Lambda_t}$ by construction. Since $\pi_{t,T}^{{\rm id}}  (h_T)1_{ \Lambda_t}+p_t^01_{\Omega\setminus\Lambda_t}\in \cP_{t,T}^{{\rm id}} (h_T)$ by definition of $\Lambda_t$, we deduce that $p_t\in \cP_{t,T}^{{\rm id}} (h_T)$. Therefore, $\cP_{t,T} (h_T)=L^0(J_{t,T}(h_T),\cF_t)$.
\end{proof}

\begin{coro}
Suppose that $\cV_{t,T}$ is   $\cF_t$-decomposable, $t\le T$, and consider a payoff  $h_T\in L^0(\R,\cF_T)$. Then, the closure in $L^0$ of  $\cP_{t,T}(h_T)$, $\cP_{t,T}^{{\rm id}}(h_T)$ and $\cP_{t,T}^{{\rm id,c}}(h_T)$ coincide with $L^0([\pi_{t,T},\infty),\cF_t)$.
\end{coro}

The natural question is whether $\cP_{t,T}^{{\rm id}} (h_T)=\cP_{t,T}(h_T)$. Actually, this is not the case in general, as shown in the following example:

\begin{example} {\rm We consider the framework of our paper between time $t=1$ and $t=2$. Suppose that $\Omega=\{\omega_i:~ i=1,2,3,4\}$, $\cF_1=\{A,A^c,\emptyset, \Omega\}$ where $A=\{\omega_1,\omega_2\}$, $A^c=\Omega\setminus A$, and $\cF_2$ is the family of all subsets of $\Omega$. We consider any probability measure $P$ on $\cF_2$ such that $P(\{\omega_i\})>0$ for all $i=1,2,3,4$. We assume that $\cV_{t,T}=\{V^1,V^2\}$ where $V^1(\omega_i)=i-1$ for $i=1,2,3,4$ and $(V^2(\omega_i))_{i=1}^4=\{-1,2,3,4\}$.  At last, we suppose that the payoff is $h(\omega_i)=i$ for $i=1,2,3,4$. Then, the minimal prices at time $t=1$ associated to $V^1,V^2$ are respectively  $p_1(V^1)=1$ and $p_1(V^2)=21_A$. Therefore,  $\cP_{1,2}(h)=L^0([1,\infty),\cF_1)\cup L^0([21_A,\infty),\cF_1)$. Then, $\pi_{1,2}(h)=1_A\notin \cP_{1,2}(h)$. On the other hand, we may see that $\cV_{t,T}^{{\rm id}} =\{V^1,V^2,V^3, V^4\}$ where $V^3=V^11_A+V^21_{A^c}$ and $V^4=V^21_A+V^11_{A^c}$. We then show that $p_1(V^3)=1_A$ and $p_1(V^4)=1+1_A$. It follows that $\pi_{1,2}^{{\rm id}}(h)=\pi_{1,2}(h)=1_A\in \cP_{1,2}^ {{\rm id}}(h)$ and $\cP_{t,T}^{{\rm id}}(h)=L^0([1_A,\infty),\cF_1)$. We conclude that $\cP_{t,T}^{{\rm id}}(h)\ne \cP_{1,2}(h)$.} $\Diamond$

\end{example}

\section{Topology defined by a semi-distance}\label{PseudoDistance}

\begin{definition} Let $E$ be a vector space. A semi-distance is a mapping $d$ defined on $E\times E$ with values in $\Reel_+$ such that the triangular inequality holds:
$$d(X,Y)\le d(X,Z)+d(Z,Y),\quad X,Y,Z\in E.$$

\end{definition}

\begin{example} \label{TopologieTauHatplus}{\rm  At time $t\le T$, we define on $ L^0( \Reel,\cF_{T})\times  L^0( \Reel,\cF_{T}) $ the pseudo-distance: $$\hat d^+_t (X,Y)=E(\esssup_{\cF_t}((X-Y)^+)\wedge  1  ),\quad  X,Y \in  L^0( \Reel,\cF_{T}).$$

Observe that only the triangle inequality is satisfied by   $d^+_t$.  In general $d^+_t(X,Y) \neq d^+_t(Y,X)$. For example, if $X+1\le Y$ a.s., then $d^+_t(X,Y)=0$ but $d^+_t(Y,X) =1$. In particular,  $d^+_t(X,Y)=0$ does not necessarily  imply that  $X=Y$ a.s.} $\Diamond$ 
 \end{example}
 
 \begin{example}\label{TopologieTau}{\rm 
Another pseudo-distance is given by  
$$d^+(X,Y)= E((X-Y)^+ \wedge 1).$$ Notice that  $d^+\le \hat d^+_t$.} $\Diamond$
\end{example} 
 
 A  pseudo-distance $d$   allows us to define a topologie on $ L^0(\Reel,\cF_{T})$. To do so, let us define, for every  $X_0\in L^0( \Reel,\cF_{T})$, the set
$$\cB_{\ep}(X_0) =\left\{X \in  L^0( \Reel,\cF_{T}): d(X_0,X )\le \ep \right\}$$ that we call ball of radius $\ep \in  \Reel^+$, centered at $X_0 \in  L^0( \Reel,\cF_{T}) $. A set $V \subseteq L^0( \Reel,\cF_T) $ is said a neighborhood of $X \in L^0( \Reel,\cF_{T})$ if there is $\ep \in (0,\infty) $ such that $\cB_{\ep}(X) \subset V$. A set $ O \subset L^0( \Reel,\cF_{T})  $ is said open if it is a neighborhood of all $X \in O$. We denote by $\cT_d$ the collection of all open sets.

\begin{lemm} The family $\cT_d$ of open sets defined from the pseudo-distance $d$  is a topology.

\end{lemm}
\begin{proof}
 It is clear that $L^0( \Reel,\cF_{T})$ is a neighborhood of all its elements, i.e.  $L^0( \Reel,\cF_{T})\in \cT_d$, and $\emptyset\in \cT_d$ by convention. Let $(O_i)_{i\in I}$ be a family of  open sets. Let $x\in \bigcup_{i\in I} O_i$, so that $x\in O_i$ for some $i\in I$. As  $O_i$ is open,   $O_i$ is a neighborhood of $x$ and, consequently, $ \bigcup_{i\in I} O_i$ is a neighborhood of $x$. 
 
 Let $(O_i)_{i\in I}$ be a finite family of  open sets. Let $x\in \bigcap_{i\in I} O_i$, so that $x\in O_i$ for every $i\in I$. So, for every $i\in I$,  there exist $\ep_i\in (0,\infty)$ such that $\cB_{\ep_i}(x)\subset O_i$. Let $\ep = \inf_{i\in I}(\ep_i)\in (0,\infty)$. We have  $\cB_{\ep}(x)\subset O_i$  for every $i\in I$. We conclude that $\bigcap_{i\in I} O_i$ is open. \end{proof}\smallskip

 In the following, we denote by $\widehat{\cT}_t$ the topology associated to the  pseudo-distance $\hat d^+_t$ given in Example \ref{TopologieTauHatplus}. Similarly, we denote by $\widehat\cB_{\ep}(x)$ the associated balls. We also denote by $\cT$ the topology defined by $d^+$ as in Example  \ref{TopologieTau} while the associated balls are just denoted by $\cB_{\ep}(x)$.

\begin{remark} \label{rem-topol} {\rm We observe several basic properties which are of interest:\smallskip

 \noindent 1)  The topology defined by the pseudo-distance  is not separated in general. Take for example $X,Y \in  L^0( \Reel,\cF_{T})$ such that $Y>X$ a.s. For every  $\ep \in  \Reel^+ $, $X - Y < 0\le  \ep$ hence $(X - Y)^+ =0 \le \ep$. So, 
 $$\hat d^+_t(X - Y) =E(\esssup_{\cF_t}(X-Y)^+ \wedge 1) \le  \ep\wedge 1$$ and we conclude that $ Y\in \widehat{\cB}_{\ep}(X)$. \smallskip
 
\noindent 2) A sequence $(X_n)_{n\in \N}$ of elements in $L^0( \Reel,\cF_{T})$ converges  to $ X \in  L^0( \Reel,\cF_{T})$ with respect to $\cT_d$ if, for all  $\ep \in \Reel^+$,  there exist $n_0 \in \N $  such that, for any $n \ge n_0$, $X_n \in \cB_{\ep}(X)$. \smallskip

\noindent 3) If $A$ is a subset of $E$, then $X$ belongs to the closure of $A$ with respect to  $\cT_d$ if and only if $X=\lim_n(X_n)$, i.e. $d(X,X^n)\to 0$, where $(X_n)_{n\in \N}$ is a sequence of elements of $A$. Indeed, this is a direct consequence of the construction of the balls from $d$.

\noindent 4) If  $(X_n)_{n\in \N}$ converges  to $ X $ with respect to $\widehat{\cT}_t$ then $(X_n)_{n\in \N}$ converges  to $ X $ with respect to $\cT$, see Examples  \ref{TopologieTauHatplus} and \ref{TopologieTau}.

\noindent 5) If  $(X_n)_{n\in \N}$ converges  to $ X $ with respect to $\widehat{\cT}_t$  and 
 $(\tilde X_n)_{n\in \N}$ is another sequence such that $\tilde X_n\ge X_n$ a.s., for all $n\in \N$, then $(\tilde X_n)_{n\in \N}$ converges  to $ X $ with respect to $\widehat{\cT}_t$.} $\Diamond$
  \end{remark}

 \begin{remark}\label{Rq6}{\rm  We recall that $d(X,Y)=E(\vert X-Y \vert \wedge 1)$  is the distance generating the convergence in probability. So, a sequence $(X_n)_{n\in \N}$ of elements in $L^0( \R,\cF_{T})$ converges  to $ X \in  L^0( \R,\cF_{T})$ with respect to $\widehat \cT_t$, see Example \ref{TopologieTauHatplus},  if and only $\esssup_{\cF_t}(X-X_n)^+$ converges to $0$ in probability. Consequently there exists a subsequence $(X_{n_k})_k$ of $(X_n)_n$ such that $\esssup_{\cF_t}(X-X_{n_k})^+$ converges to $0$ almost surely, i.e.  for every $\ep \in \Reel^+$ there exists $k_0$ such that, for all $k>k_0$, we have $\esssup_{\cF_t}(X-X_{n_k})^+ \le \ep $, which implies that $X \le \ep+ X_{n_k}$.} $\Diamond$

\end{remark}

\begin{lemm}
If $F$ is  a  closed set for $\cT$ (resp. for $\widehat{\cT}_t$),  then $F$ is a lower set, i.e. $F-L^0(\Reel_+,\cF_T)\subseteq F$. \end{lemm}
\proof

Indeed, consider $Z\le \gamma$ where $\gamma \in F$. Then, $(Z-\gamma)^+=0$ hence the constant sequence $(\gamma_n=\gamma)_{n\ge 1}$ converges to $Z$ and, finally, $Z\in F$. Note that, if $F$ is closed for $\cT$, it is closed for  $\widehat{\cT}_t$.
\endproof

 \begin{lemm} \label{linear} Let $d$ be a pseudo-distance on $E\times E$. Consider two sequences $(X_n)_{n\in \N}$ and $(Y_n)_{n\in \N}$ of elements in $E$ which  converge   to $ X , Y \in  L^0( \R,\cF_{T})$  respectively with respect to $\cT_d$.  If $d(a+b,a+c)\le d(b,c)$ for all $a,b,c\in E$, then  $(X_n+Y_n)_{n\in \N}$ converges to $X+Y$.\smallskip
 \end{lemm}
 \begin{proof} It suffices to observe that 
  \bean
  d(X+Y,X_n+Y_n)&\le&d(X+Y,X_n+Y)+d(X_n+Y,X_n+Y_n)\\
  &\le& d(X,X_n)+d(Y,Y_n).
  \eean
  \end{proof}

\begin{prop}\label{linear+} Consider the  pseudo-distance  $\hat d^+_t$ from Example \ref{TopologieTauHatplus}. Let  $(X_n)_{n\in \N}$ and $(Y_n)_{n\in \N}$ be two sequences of elements in $L^0( \R,\cF_{T})$ which  converge respectively  to $ X , Y \in  L^0( \R,\cF_{T})$ with respect to $\widehat \cT_t$. The following convergences hold with respect to $\widehat \cT_t$:

\begin{itemize}

\item [1)] The sequence $(\alpha_t X_n)_{n\in \N}$  converges  to $\alpha_t X$, for all $\alpha_t \in L^0(\Reel_+,\cF_t)$.\smallskip
 
\item [2)] The sequence $(\alpha X_n)_{n\in \N}$  converges  to $\alpha X $,  for all $\alpha \in L^{\infty}(\Reel_+,\cF_T)$.\smallskip

\item[3)] The sequence  $(\esssup_{\cF_t}(X_n))_{n\ge 1}$ converges to $\esssup_{\cF_t}(X)$. 

 \end{itemize}
 
 Moreover, the two first statements remain true if we replace $\widehat \cT_t$ by $\cT$.

\end{prop}

\begin{proof}
 Recall that $ \esssup_{\cF_t}(\alpha_t X- \alpha_t X_n)^+ = \alpha_t \esssup_{\cF_t}(X-X_n)^+$ if $\alpha_t$ belongs to $L^0(\Reel_+,\cF_t)$.  Then, for all $\gamma>0$,
 \bean d^+_t(\alpha_t X,\alpha_t X_n)&=& E(\alpha_t \esssup_{\cF_t}(X-X_n)^+\wedge 1. 1_{\esssup_{\cF_t}(X-X_n)^+< \gamma})\\
 &&+E(\alpha_t \esssup_{\cF_t}(X-X_n)^+\wedge 1. 1_{\esssup_{\cF_t}(X-X_n)^+\ge \gamma})\\
 &\le& E(\alpha_t \gamma\wedge 1)+P(\esssup_{\cF_t}(X-X_n)^+\ge \gamma).
 \eean
 By the dominated convergence theorem, we may fix $\gamma$ small enough such that $E(\alpha_t \gamma\wedge 1)\le \epsilon/2$, where $\epsilon>0$ is arbitrarily chosen. Moreover, by assumption, $P(\esssup_{\cF_t}(X-X_n)^+\ge \gamma)\le  \epsilon/2$, if $n$ is large enough. We get  that $d^+_t(\alpha_t X,\alpha_t X_n)\le \epsilon$, if  $n$ is large enough, i.e. $\alpha_t X_n\to \alpha_t X$.\smallskip
  
 The second statement is a consequence of the first one as we may observe that, for all $\alpha \in L^{\infty}(\Reel_+,\cF_T)$,   
 $$d^+_t(\alpha X,\alpha X_n) \le d^+(\|\alpha\|_{\infty} X,\|\alpha\|_{\infty} X_n).$$ 
 
 At last,  notice that the following inequality holds
$$ \esssup_{\cF_t}(X)=\esssup_{\cF_t}(X+X_n-X_n)  \le  \esssup_{\cF_t}(X-X_n)+\esssup_{\cF_t}(X_n) .$$
Therefore,
\bean
&&\esssup_{\cF_t}(X) - \esssup_{\cF_t}(X_n)  \le  \esssup_{\cF_t}(X-X_n)^+,\\
&&\esssup_{\cF_t} ((\esssup_{\cF_t}(X) - \esssup_{\cF_t}(X_n))^+ ) \le  \esssup_{\cF_t}(X-X_n)^+, \\ 
&&d^+_t(\esssup_{\cF_t}(X),\esssup_{\cF_t}(X_n)^+) \le  E(\esssup_{\cF_t}((X-X_n)^+) \wedge 1). 
\eean
The conclusion follows.
\end{proof}

\begin{remark}{\rm 
If a sequence $(X_n)_n$ converges to $X$ with respect to $\widehat{\cT}$ or $\cT$ it does not imply that $(-X_n)_n$ converges to $-X$.
Take for example the sequence $(-1)^n$. We have $(-1-(-1)^n)^+=0$ for any $n \in \N$. Then, $(-1)^n$ converges to $-1$ for $\widehat{\cT}$ and  $\cT$. But $(1-(-1)^{n+1})^+ \wedge 1 = 1$ when $n$ is even. Then $(1-(-1)^{n+1})^+$ does not converge to $0$ in probability. So, $-(-1)^n$  does not converge to $-1$ for $\cT$ nor for $\widehat{\cT}$.} $\Diamond$
\end{remark}

 \begin{lemm}\label{conv-subseq}
 Let $(X_n)_{n\in \N}$ be a sequence of elements in $L^0( \R,\cF_{T})$ that converge  to $ X \in  L^0( \R,\cF_{T})$  with respect to $ \cT$.  Then, for  every random subsequence $(n_k)_{k\ge 1}$, $(X_{n_k})_k$  converges  to $ X$ with respect to $\cT$. The same holds with respect to $\widehat{\cT}_t$ if the  random subsequence $(n_k)_{k\ge 1}$ is $\cF_t$-measurable.

 \end{lemm}
 \begin{proof}
 Note that  $(X-X_{n_k})^+ = \sum_{j=k}^\infty (X-X_j)^+ 1_{n_k = j}.$ Therefore, 
 
  \bean \mathbb{P}((X-X_{n_k})^+ \ge \ep)&=& \mathbb{P} (\sum_{j=k}^\infty \{(X-X_j)^+  \ge \ep \} \cap \{ n_k = j \}),\\
  &\le&  \sum_{j=k}^\infty  \mathbb{P}\left( \{(X-X_j)^+  \ge \ep \} \cap \{ n_k = j \}\right) .
  \eean
 Let $\alpha > 0$.
 Consider M such that $\sum_{j=M+1}^\infty \mathbb{P}(n_k=j) \le \alpha/2$  and $k_0$  such that, for every $k \ge k_0$,  we have $\mathbb{P}((X-X_{k})^+ \ge \ep)\le \alpha /2M $. Then,
 \bean \mathbb{P}((X-X_{n_k})^+ \ge \ep) &\le&  \sum_{j=k}^{M\vee k}  \mathbb{P}( \{(X-X_j)^+ \ge \ep) +  \sum_{j=M+1}^\infty  \mathbb{P}(n_k =j) \\
 &\le& M \alpha/2M + \alpha/2 \le \alpha .
 \eean
 So $(X-X_{n_k})^+$ converges to zero in probability hence $(X_{n_k})_k$  converges  to $ X$ with respect to $\cT$.
 
 For the second statement, it suffices to observe that, when $(n_k)_{k\ge 1}$ is $\cF_t$-measurable, we have:
 \bean (X-X_{n_k})^+&\le& \sum_{j=k}^\infty\esssup_{\cF_t} (X-X_j)^+ 1_{n_k = j},\\
 \esssup_{\cF_t} (X-X_{n_k})^+&\le& \sum_{j=k}^\infty\esssup_{\cF_t} (X-X_j)^+ 1_{n_k = j}.
 \eean
 It is then possible to repeat the previous reasoning, replacing  $(X-X_j)^+$ by $\esssup_{\cF_t} (X-X_j)^+ $, $j\ge 1$. \end{proof}
\begin{prop} \label{CriteriaConv}

 A sequence $(X_n)_{n\in \N}$ of elements in $L^0( \R,\cF_{T})$ converges  with respect to $\widehat \cT_t$ (respectively $\cT$) if and only if 
$$  \inf_n( X _{n}) > -\infty.$$ 
Moreover, $\inf_n( X _{n})$ is a limit of $(X_n)_{n\in \N}$  for  $\widehat \cT_t$ and $ \cT_t$.

\end{prop}
 \begin{proof}
 Suppose that $(X_n)_{n\in \N}$ converges  to $ X $ with respect to $\cT$ and suppose that $  \inf_n( X _{n}) = -\infty$ on a non null set. Then, on this set, there exists a random subsequence $X_{n_k}$ that converges to $-\infty$ almost surely.  By Lemma \ref{conv-subseq}, $(X_{n_k})_{n\in \N}$ converges  to $ X $ with respect to $\cT$. In other words, $(X-X_{n_k})^+$ converges to zero in probability. Therefore, there exits a subsequence $X_{n_{k_j}}$ such that $(X-X_{n_{k_j}})^+$ converges to zero almost surely. This is in contradiction with  the fact that $X_{n_{k_j}}$ converges to $-\infty$.
 
 Now suppose that  $  \inf_n( X _{n}) > -\infty$. We have  $ X_n \ge \inf_n( X _{n}) > -\infty$. So $( \inf_n( X _{n})- X_n)^+=0$. This implies that $\esssup_{\cF_t}( \inf_n( X _{n})- X_n)^+=0$ hence $(X_n)_{n\ge 1}$ converges to $ \inf_n( X _{n})$ with respect to $\widehat \cT_t$.  
 \end{proof}
  
 \begin{coro}
  A sequence $(X_n)_{n\in \N}$ of elements in $L^0( \R,\cF_{T})$ is such that $(X_n)_{n\in \N}$  and $(-X_n)_{n\in \N}$ converge   with respect to $\widehat \cT_t$ (respectively $\cT$) if and only if $\sup_n(\vert X_n\vert ) < \infty$ almost surely. 
 \end{coro}

 \begin{coro}
 A sequence $(X_n)_{n\in \N}$ of elements in $L^0( \R,\cF_{T})$ converges  with respect to $\widehat \cT_t$ if and only if  $(X_n)_{n\in \N}$  converges  with respect to $\cT$( not necessarily with the same limits). 
 \end{coro}
 \begin{lemm}
 A sequence $(X_n)_{n\in \N}$ of elements in $L^0( \R,\cF_{T})$ is such that $(X_n)_{n\in \N}$ converges to $X$ and $(-X_n)_{n\in \N}$ converges to $-X$ with respect to $\widehat \cT_t$ if and only if  $\esssup_{F_t}(\vert X-X_n \vert)$ converges to $0$ in probability.
 \end{lemm}
 
 \begin{prop}\label{liminf}
If a sequence $(X_n)_{n\in \N}$ of elements in $L^0( \R,\cF_{T})$ converges  to $ X \in  L^0( \R,\cF_{T})$, with respect to $\widehat \cT_t$ (resp. $\cT$), then there exists a deterministic subsequence $(n_k)_{k\ge 1}$ such that  $$ X \le \liminf_k( X _{n_k}).$$ 
\end{prop}
\begin{proof}
Recall that a sequence $(X_n)_{n\in \N}$ of elements in $L^0( \R,\cF_{T})$ converges  to $ X \in  L^0( \R,\cF_{T})$  if and only if $\esssup_{\cF_t}(X-X_n)^+$ converges to $0$ in probability. Therefore, there exists a subsequence $(n_k)_{k\ge 1}$  such that $\esssup_{\cF_t}(X-X_{n_k})^+$ converges to $0$ almost surely. As 
$$ X-X_{n_k} \le \esssup_{\cF_t}(X-X_{n_k})^+$$
then $ \liminf_k [X - \esssup_{\cF_t}(X-X_{n_k})^+] \le \liminf_k( X _{n_k})$. So, we deduce that 
$$ X \le \liminf_k( X _{n_k}).$$ The same reasoning holds for $\cT$. \end{proof} 

\begin{definition}
 For a converging sequence $X=(X_n)_n$ we denote by $\widehat{\cL}(X)$ (resp. $\cL(X)$  )  the set of all limits with respect to $\widehat \cT_t$ and  $
  \cT_t$ respectively. 
 \end{definition}

\begin{lemm}\label{LemmConA.S.Tau}
 If a sequence $(X_n)_n$ converges to $X$ in probability then $(X_n)_n$ converges to $X$ for the topology $\cT$ and $\cL(X)=L^0((-\infty,X],\cF_T)$.

\end{lemm}

\begin{proof}
If $\vert X_n - X \vert $ converges to zero in probability then  the same holds for  $(X_n-X)^+$. Indeed, $(X_n-X)^+ \le  \vert  X_n - X  \vert $.  Therefore, $(X_n)_n$ converges to $X$ for the topology $\cT$. Moreover, there exists a subsequence $(n_k)_{k\ge 1}$ such that $(X_{n_k})_{k\ge 1}$  converges to $X$ a.s. but also in $\cT$ by the first part. By  Proposition \ref{liminf}, any $Z\in \cL(X)$ satisfies $Z\le X$. The conclusion follows. \end{proof}
\begin{remark} \label{ctr.ex}{\rm 
The convergence almost surely to a limit $X$ does not imply the convergence for $\widehat{\cT}$ to $X$. Also the convergence for $\widehat{\cT}$ and $\cT$ does not necessarily imply the almost surely convergence. To see it, let us consider the two following examples.
\begin{itemize}

\item [1)] We consider $\Omega=[0,1]$ equipped with the Lebesgue measure. Take the sequence $X_n (\omega)= -1$ on $[0, 1/n]$ and $X_n (\omega)= 1/2^n$  on $(1/n, 1]$, $n\ge 1$. It is clear that $(X_n)_n$ converges to $X_0=0$ almost surely. But observe that
$\esssup_{\cF_0}(X_0-X_n)^+ = 1$. So, $X_n$ does not converge to $0$ for $\widehat{\cT}_0$. Note that $X_n$ converges to $-1$ for $\widehat{\cT}_0$ and $\cT$. 
\item [2)] We consider $\Omega=\Reel_+$ equipped with the Lebesgue measure. Consider $X_n (\omega) = cos(n\omega)$ for any $\omega \in \Reel$ and  $Y_n(\omega) = (-1)^n$, $n\ge 0$. Then,  $(X_n)_n$ and $(Y_n)_n$ do not converge almost surely but 
$(X_n)_n$ and $(Y_n)_n$ converge for $\cT$ and $\widehat{\cT}$ towards $-1$.  $\Diamond$

\end{itemize}}

\end{remark}
\begin{definition} [Cauchy sequence]

A sequence $(X_n)_n$ is said a Cauchy sequence for the pseudo-distance $d$ if :
$$ \forall \ep >0, \exists n_0 , \forall n,m \ge n_0,   d(X_n,X_m) \le \ep.$$

\end{definition}
\begin{remark}
If a sequence $(X_n)_n$ is convergent for  $\widehat{\cT}$ (or $\cT$) it is not necessarily a Cauchy sequence. 
Take the  sequence $X_n = (-1)^n$. It converges but it is not a Cauchy one. In fact
$$ d_t^+(X_{2n} ,X_{2n+1}) =1 , \forall n \in \N. \,\,\Diamond$$

\end{remark}
\begin{prop}
Every Cauchy sequence for  $d_t^+$  is convergent in probability.
\end{prop}
\begin{proof}
Let $(X_n)_n$ be a Cauchy sequence for $d_t^+$:
 $$\forall \ep >0, \exists n_0 , \forall n,m \ge n_0,   d_t^+(X_n,X_m) \le \ep. $$
 So, we also have  $ d_t^+(X_m,X_n) \le \ep$. In other terms $E((X_n-X_m)^+ \wedge 1) \le \ep$ and $E((X_m-X_n)^+ \wedge 1) \le \ep$. Then $E(\vert X_n-X_m\vert \wedge 1) \le \ep$. Then $(X_n)_n$ is a Cauchy sequence for the convergence in probability. Consequently it is convergent for the convergence in probability. 
\end{proof}

\begin{example}\label{constant}{\rm

 Let $C \in \Reel$. Consider the sequence $X=(X_n)_n$ of elements in $L^0( \R,\cF_{T})$ such that $X_n=C$ for every $n\in \N$. Consider any $Z\in \widehat\cL(X)$. By  Proposition \ref{liminf}, $Z\le C$. On the other hand, $(C-X_n)^+=0$ hence $(X_n)$ converges to $C$  in $\widehat \cT_t$. By similar arguments, we finally deduce that $\widehat \cL(X) =\cL(X)=L^0 ( (-\infty ,C],\cF_T)$.} $\Diamond$

\end{example}

\begin{example} \label{ex2}{\rm 
Consider the sequence $X=(X_n)_n$ of elements in $L^0( \R,\cF_{T})$ such that $X_n=(-1)^n$ for every $n\in \N$. We have  $\cL(X)=L^0 ( (-\infty ,-1],\cF_T)$. Indeed, as $E[ (-1 -(-1)^n) \wedge1] =0$, $-1$ is a limit of $X$ for $\cT$.  So for any $Z\le -1$, $Z$ is a limit for $X$. Now consider any $Z\in \cL(X)$. Let us show that, $Z\le -1$.
We know that $(Z- (-1)^n)^+$ converges to zero in probability. Then, if $A_n= \{ (Z-(-1)^n)^+ \le \ep\}$, $\mathbb{P} ( A_n) $ converges to 1 when $n \rightarrow \infty$. On $A_n$,  $Z-(-1)^n \le \ep$ hence $ Z \le \ep - 1$ when $n$ is odd. As $n$ goes to $\infty$ we deduce that $ Z \le \ep - 1$ almost surely. To see it,  suppose by contradiction that $\mathbb{P}(B) >0$ where $B= \{ Z > \ep -1 \}$. Therefore, there exists $n_0$ such that $\mathbb{P}(B \cap A_n) >0$ for any $ n \ge n_0$. If not, there exists a subsequence $(A_{n_k})$ such that $\mathbb{P}(B \cap A_{n_k}) =0$. Hence, $\mathbb{P}(A_{n_k}) =\mathbb{P}(B^c \cap A_{n_k})  \le \mathbb{P}(B^c) < 1$, in contradiction with $\lim_{k\to \infty}\mathbb{P} ( A_{n_k}) =1$. Finally, $\mathbb{P}(B \cap A_n) >0$ for any $ n \ge n_0$ in contradiction with the inequality  $ Z \le \ep - 1$ on $A_n$, when $n$ is odd. We conclude that  $ Z \le \ep - 1$ a.s. and the result follows.  We also deduce that $\widehat \cL(X)=\cL(X)$.} $\Diamond$
\end{example}

\begin{example} \label{ex3}{\rm 
Consider the sequence $X=(X_n)_n$ of elements in $L^0( [0,1],\cF_{T})$, equipped with the Lebesgue measure, such that $X_n(\omega)=-1_{[0,1/n]}$ for every $n\ge 1$. We suppose that $\cF_0$ is trivial. We know by Lemma \ref{LemmConA.S.Tau} that  $\cL(X)=L^0 ( (-\infty ,0],\cF_T)$ but $\widehat\cL(X) \subset L^0 ( (-\infty ,0],\cF_T)$. Indeed, $0$ is not  a limit for $\widehat{ \cT_0}$ as $\esssup_{\cF_0}(0-X_n)^+=1$.

Moreover, consider $\widehat X_{\infty}\in \widehat\cL(X) $. Observe that  the deterministic sequence $\alpha_n=\esssup_{\cF_0}(\widehat X_{\infty}-X_n)^+$ converges to $0$ and $\widehat X_{\infty}-X_n\le (\widehat X_{\infty}-X_n)^+\le \alpha_n$. We finally conclude that $\widehat\cL(X)$ is the family of all random variables $\widehat X_{\infty}$ that satisfies $\widehat X_{\infty}\le \inf_n(X_n+\alpha_n)$ for some non negative deterministic sequence $(\alpha_n)_{n\ge 1}$ with $\lim_{n\to \infty}\alpha_n=0$. For example, take $\alpha_n=1$ if $n<n_0$, $n_0>0$ is fixed, and $\alpha_n=0$ otherwise. Then,  $Z_{n_0}= \inf_{n \ge n_0} X_n \in \widehat\cL(X)$.} $\Diamond$
\end{example}

  \begin{prop} \label{set-of-limits} If a sequence $X=(X_n)_n$ of elements in $L^0( \Reel,\cF_{T})$ converges in $\widehat\cT$, then the set 
  $\widehat\cL(X)$ coincides with the family of all $\widehat X_{\infty}$ such that $\widehat X_{\infty}\le \inf_n(X_n+\alpha_n)$ for some sequence $(\alpha_n)_{n\ge 1}$ in $L^0( \Reel_+,\cF_{t})$ that converges to zero in probability. 
  If a sequence $X=(X_n)_n$ of elements in $L^0( \Reel,\cF_{T})$ converges in $\cT$, then the set 
  $\cL(X)$ coincides with the family of all $ X_{\infty}$ such that $ X_{\infty}\le \inf_n(X_n+\alpha_n)$ for some sequence $(\alpha_n)_{n\ge 1}$ in $L^0( \Reel_+,\cF_{T})$ that converges to zero in probability. 
  
  \end{prop}
  \begin{proof} Consider a sequence $X=(X_n)_n$ of elements in $L^0( \Reel,\cF_{T})$ converging for $\widehat\cT$.
  Let $\widehat X_{\infty}\in \widehat\cL(X) $. By definition, $\alpha_n=\esssup_{\cF_t}(\widehat X_{\infty}-X_n)^+$ converges to $0$ in probability. As $\widehat X_{\infty}-X_n\le \esssup_{\cF_t}(\widehat X_{\infty}-X_n)^+\le \alpha_n$, then we deduce that $\widehat X_{\infty}\le \inf_n(X_n+\alpha_n)$. Conversely, if $\widehat X_{\infty}\le \inf_n(X_n+\alpha_n)$, then $\widehat X_{\infty}\le X_n+\alpha_n$. Therefore, $\esssup_{\cF_t}(\widehat X_{\infty}-X_n)^+\le \alpha_n$ and the conclusion follows.
  For the second statement it suffices to consider $\alpha_n=( X_{\infty}-X_n)^+$. \end{proof}

\section*{Declarations}

Not applicable

\begin{appendices}

\section{Proof of Proposition \ref{S-SM}}\label{secA1}

 The proof of Proposition \ref{S-SM} is deduced from Lemma \ref{maxingale5}. To get it, we first show intermediate steps such as the following Lemma \ref{maxingale1}, Lemma \ref{maxingale2}, Lemma \ref{maxingale3} and Lemma \ref{maxingale4}.
 \begin{lemm}\label{maxingale1}
Let $(M_t)_{t\in [0,T]}$ be  a sub-maxingale. Let $ \tau$ be a stopping time such that $\tau(\Omega)= \{ t_1, t_2, \cdots,t_n \}$ where $(t_i)_{i=1}^n$ is an increasing sequence of discrete dates. Then, for all $i=1,\cdots,n$, we have  
 $ \esssup_{\cF_{t_i}}(M_{\tau }) \ge M_{\tau \wedge t_i}$.

\end{lemm}
\begin{proof}  We have:

\bean \esssup_{\cF_{t_i}}(M_{\tau \wedge t_{i+1}}) 
&=& \esssup_{\cF_{t_i}}(M_{\tau \wedge t_{i+1}} 1_{\{\tau  \le t_i \}})+\esssup_{\cF_{t_i}}(M_{\tau \wedge t_{i+1}} 1_{\{\tau >t_i \}}),\\
&=&1_{\{\tau >t_i \}}\esssup_{\cF_{t_i}}(M_{ t_{i+1}} )+ 1_{\{\tau  \le t_i \}} \esssup_{\cF_{t_i}}(M_{\tau \wedge t_{i}} ),\\
&\ge &1_{\{\tau >t_i \}}M_{ t_{i}} + 1_{\{\tau  \le t_i \}}  M_{\tau \wedge t_{i}} 
= M_{\tau \wedge t_i}.\eean

If $j >i+1$, argue by induction. By the tower property, we first have $\esssup_{\cF_{t_i}}(M_{\tau \wedge t_{j}}) =\esssup_{\cF_{t_i}}(\esssup_{\cF_{t_{j-1}}}(M_{\tau \wedge t_{j}}) )$.  Therefore, by the first step above,  $\esssup_{\cF_{t_i}}(M_{\tau \wedge t_{j}})\ge \esssup_{\cF_{t_i}}(M_{\tau \wedge t_{j-1}})$ and we conclude by induction. 
\end{proof}

\begin{lemm}\label{maxingale2}
Let $ \tau$ be a stopping time such that $\tau(\Omega)= \{ t_1, t_2, \cdots,t_n \}$ where $(t_i)_{i=1}^n$ is an increasing sequence of discrete dates. Then, for any random variable $X$, we have
$$\esssup_{\cF_{\tau}}(X 1_{\{\tau = t_i\}})= \esssup_{\cF_{t_i}}(X )1_{\{\tau = t_i\}}.$$
\end{lemm}
\begin{proof}
As $1_{\{\tau = t_i\}}$ is $\cF_{\tau}$-mesurable, then we get that
$$\esssup_{\cF_{\tau}}(X1_{\{\tau = t_i\}})=\esssup_{\cF_{\tau}}(X ) 1_{\{\tau = t_i\}}.$$
 Since $X 1_{\{\tau = t_i\}}\le \esssup_{\cF_{t_i}}(X )1_{\{\tau = t_i\}}$, we deduce that  
 $$\esssup_{\cF_{\tau}}(X 1_{\{\tau = t_i\}})\le \ \esssup_{\cF_{\tau}}(\esssup_{\cF_{t_i}}(X )1_{\{\tau = t_i\}}).$$
 We  claim that $Z=\esssup_{\cF_{t_i}}(X )1_{\{\tau = t_i\}}$ is $\cF_{\tau}$-mesurable. For any $k \in \Reel$, 
 $$\{Z \le k \} = \{0 \le k \} \cap \{ \tau \neq t_i \} \cup \{ \tau = t_i \} \cap \{ \esssup_{\cF_{t_i}}(X) \le k  \}.$$
 Note that 
  $\{0 \le k \} = \emptyset$ or $\Omega$ and 
   $\{ \tau \neq t_i \} \in \cF_{\tau}$ hence $\{0 \le k \} \cap \{ \tau \neq t_i \}\in \cF_{\tau}$.
   Now let us show that 
   $B = \{ \tau = t_i \} \cap \{ \esssup_{\cF_{t_i}}(X) \le k  \} \in \cF_{\tau}$. To do  so, we evaluate
   $B \cap \{ \tau \le t \}$ for $t\ge 0$. Note  that $t_j \le  t < t_{j+1}$ for some $t_j \in \{t_0,\cdots,t_n,t_{n+1}\}$, where $t_{n+1}=\infty$. So, we deduce that $B \cap \{ \tau \le t \}$ coincides with $B \cap \{ \tau \le t_j \}= \emptyset $ if $t_j < t_i$. Otherwise, we obtain  that $B \cap \{ \tau \le t \}=B \in F_{t_i} \subseteq F_{t_j} \subseteq F_t.$ Therefore, $B \cap \{ \tau \le t \}\in \cF_t$, for all $t\in \Reel$, hence  $B \in \cF_{\tau}$. Finally, $Z$ is $\cF_{\tau}$-mesurable and the inequality $\esssup_{\cF_{\tau}}(X 1_{\{\tau = t_i\}})\le \ \esssup_{\cF_{t_i}}(X )1_{\{\tau = t_i\}} $ holds. 
   For the reverse inequality it suffices to show that $Y=\esssup_{\cF_{\tau}}(X 1_{\{\tau = t_i\}} )$ is $\cF_{t_i}$-measurable. Since $\{\tau \ne t_i\}\in \cF_{\tau}$, we get that $Y 1_{\{\tau \ne t_i\}}=0$ and  $$ \{ Y \le k \} = (\{0 \le k \} \cap \{\tau \neq t_i \} )\cup (A \cap \{\tau = t_i \}),$$
   with $A=\{\esssup_{\cF_{\tau}}(X 1_{\{\tau = t_i\}}) \le k \} $. As $A \in \cF_{\tau}$,  $ A \cap \{\tau \le  t_i \} \in \cF_{t_i}$ and, finally, $A \cap \{\tau = t_i \}= A \cap \{\tau = t_i \} \cap \{\tau \le t_i \}\in \cF_{t_i}$.  Therefore, for all $k\in \Reel$, $ \{ Y \le k \}\in \cF_{t_i}$, i.e.  $Y$ is $\cF_{t_i}$-measurable. At last, notice that   $\esssup_{\cF_{\tau}}(X 1_{\{\tau = t_i\}} )\ge X 1_{\{\tau = t_i\}} $ and, since $Y$ is $\cF_{t_i}$-measurable, we get that  $\esssup_{\cF_{\tau}}(X 1_{\{\tau = t_i\}} )\ge \esssup_{\cF_{t_i}}(X 1_{\{\tau = t_i\}})$. The conclusion follows. \end{proof}
\begin{lemm}\label{maxingale3}
Let $(M_t)_{t\in [0,T]}$ be a sub-maxingale. Let $ \tau$, $S$ be two  stopping times. Suppose  that $S(\Omega)= \{ t_1, t_2, \cdots,t_n \}$ where $(t_i)_{i=1}^n$ is an increasing sequence of discrete dates and suppose that $\tau(\Omega)$  is also a finite set. Then  $ \esssup_{\cF_{S}}(M_{\tau })\ge  M_{\tau \wedge s}$.

\end{lemm}
\begin{proof}
By lemma \ref{maxingale2}, we obtain
 $ \esssup_{\cF_{S}}(M_{\tau }) 
=\sum\limits_{i=1}^n  \esssup_{\cF_{t_i}}(M_{\tau  }) 1_{\{S=t_i\}}$.
 By lemma \ref{maxingale1}, we deduce that 
 \bean \esssup_{\cF_{S}}(M_{\tau })
\ge \sum_{i=1}^n M_{\tau \wedge t_i} 1_{\{S=t_i\}}= \sum_{i=1}^n M_{ \tau \wedge S} 1_{\{S=t_i\}}
= M_{ \tau \wedge S}.\eean\end{proof}

\begin{lemm}\label{maxingale4}
Let $\tau \in [0,T]$ be a stopping time. Suppose that the filtration $(\cF_t)_{t\in [0,T]}$ is right-continuous. There exists a non increasing sequence 
 $(\tau_n)_n$ of stopping times converging to $\tau$ such that, for any $X\in L^0(\R,\cF_T)$,
 $$ \esssup_{\cF_{\tau}}(X)= \lim_n \uparrow  \esssup_{\cF_{\tau_n}}(X).$$
 Moreover, $\tau^n(\Omega)$ is finite for all $n\ge 1$.
\end{lemm}
\begin{proof}

Let $\tau$ be a stopping time taking values in $[0,T]$. For any $n\ge 1$, we define $\tau^n(\omega) = T(i+1)/2^n$ where $i=i(\omega)$ is uniquely defined such that $Ti/2^n< \tau(\omega) \le T(i+1)/2^n$ for $i\ge 1$ or $0\le  \tau(\omega) \le T/2^n$ when $i=0$. Note that $\tau^n(\Omega)$ is finite  and $\tau^n \ge \tau$.  It is easily seen that $(\tau_n)_n$ is non increasing, positive and $\lim_n \tau_n=\tau$. Moreover, $\tau^n$ is a stopping time. Indeed, for any fixed $t \in [0,T)$, there exists $i \in \N$ such that $Ti/2^n\le t < T(i+1)/2^n$. Then $\{ \tau^n \le t \}= \{ \tau \le Ti/2^n \} \in \cF_{Ti/2^n} \subset \cF_t$ and the conclusion follows. 

As $(\tau^n)_n$ is non increasing,  then $(\cF_{\tau^n})_n$  non increasing. As we know that $\esssup_{\cF_{\tau^{n+1}}}(X) \ge X$ and $\esssup_{\cF_{\tau^{n+1}}}(X) $ is $\cF_{\tau^n}$-measurable ($\tau_{n+1}\le \tau_n)$, we deduce  that $\esssup_{\cF_{\tau^n}}(X) \le \esssup_{\cF_{\tau^{(n+1)}}}(X) $, i.e.  $(\esssup_{\cF_{\tau^n}}(X))_n$ is non decreasing. 

Similarly,   $\tau^n \ge \tau$ implies that $\esssup_{\cF_{\tau^n}}(X) \le \esssup_{\cF_{\tau}}(X) $. Therefore,
$\lim_n \uparrow  \esssup_{\cF_{\tau^n}}(X) \le \esssup_{\cF_{\tau}}(X) $. To obtain the reverse inequality, we consider the sequence $(\esssup_{\cF_{\tau+ T/n}}(X))_n.$   Since  $\tau+ T/n \ge \tau^n$, then
$$\lim_n \uparrow  \esssup_{\cF_{\tau + T/n}}(X) \le \lim \uparrow  \esssup_{\cF_{\tau^n}}(X) \le \esssup_{\cF_{\tau}}(X) .$$
It suffices to see that $Z=\lim_n \uparrow  \esssup_{\cF_{\tau + T/n}}(X)$ is $\cF_{\tau}$-measurable to conclude. Indeed, $Z\ge X$ hence $Z\ge  \esssup_{\cF_{\tau}}(X) $ and  inequalities above are equalities. For all $k\in \Reel$, $t\ge 0$, and any $n_0\ge 1$
\bean \{Z\le k\}\cap \{\tau\le t\}&=&\bigcap_{n\ge 1}\{ \esssup_{\cF_{\tau + T/n}}(X)\le k\}\cap\{\tau\le t\},\\
&=&\bigcap_{n\ge n_0}\{ \esssup_{\cF_{\tau + T/n}}(X)\le k\}\cap \{\tau+ T/n\le t+ T/n\}.
\eean

Notice that $\esssup_{\cF_{\tau + T/n}}(X)$ is $\cF_{\tau+T/n}$-measurable. We deduce that:

\bean &&\{ \esssup_{\cF_{\tau + T/n}}(X)\le k\}\in \cF_{\tau+T/n},\\
&&\{ \esssup_{\cF_{\tau + T/n}}(X)\le k\}\cap \{\tau+ T/n\le t+ T/n\}\in \cF_{t+ T/n}.\eean
Therefore, for any $\epsilon>0$ and $n_0\ge 1$ such that $t+ T/n\le t+\epsilon$, we have $\cF_{t+ T/n}\subseteq \cF_{t+\epsilon}$ and, finally, $\{Z\le k\}\cap \{\tau\le t\}\in \cap_{\epsilon>0}\cF_{t+\epsilon}=\cF_{t+}=\cF_t$. We deduce that $\{Z\le k\}\in \cF_{\tau}$, for all $k\in \Reel$, i.e. $Z$ is $\cF_{\tau}$-measurable. \end{proof}
 \begin{lemm}\label{maxingale5} Suppose that the filtration $(\cF_t)_{t\in [0,T]}$ is right-continuous.
Let $(M_t)_{t\in [0,T]}$ be a  right-continuous sub-maxingale. Let $\tau$, $S$ be two  stopping times such that   $\tau(\Omega)$  is a finite set. Then,  we have $ \esssup_{\cF_{S}}(M_{\tau })\ge M_{\tau \wedge S}$.

\end{lemm}
\begin{proof}
Let $(S_n)_n$ be a  sequence of stopping times decreasing to $S$ as given in Lemma \ref{maxingale4}. Recall that $S_n(\Omega)$ is finite for all $n$. Moreover, we have  
 $ \esssup_{\cF_{s}}(M_{\tau })= \lim_n \uparrow  \esssup_{\cF_{s_n}}(M_{\tau })$. By Lemma \ref{maxingale3}, we deduce that  $ \esssup_{\cF_{S}}(M_{\tau }) \ge  \lim \uparrow  M_{\tau \wedge S_n} .$
  As $(\tau \wedge S_n)_n$ decreases to $\tau \wedge S$ and $M$ is  right-continuous,  we conclude that $\esssup_{\cF_{s}}(M_{\tau }) \ge M_{\tau \wedge s}$. \end{proof}

\section{Auxiliary results}

\begin{lemm} Suppose that, at time $t\le T$, $\cO_t$ is the pseudo-distance topology defined by (\ref{PDIST}) and $\cV_{t,T}^c=\cV_{t,T}^c(\cO_t)$. If $\cV_{t,T}$ is $\cF_t$-decomposable, then $\cV_{t,T}^c$ is $\cF_t$-decomposable.
\end{lemm}
\begin{proof} Consider $V_{t,T}^{c,i}\in \cV_{t,T}^c$, $i=1,2$, and $F_t\in \cF_t$. By Proposition \ref{set-of-limits}, $V_{t,T}^{c,i}\le V_{t,T}^{n,i}+\alpha_t^{n,i}$ where $V_{t,T}^{n,i}\in \cV_{t,T}$ and  $\alpha_t^{n,i}$ converges to $0$ in probability as $n\to \infty$, for $i=1,2$. We set 
\bean V_{t,T}^{n}&=&V_{t,T}^{n,1}1_{F_t}+V_{t,T}^{n,2}1_{\Omega\setminus F_t},\quad \alpha_t^{n}=\alpha_t^{n,1}1_{F_t}+\alpha_t^{n,2}1_{\Omega\setminus F_t}.
\eean
Note that $V_{t,T}^{n}\in   \cV_{t,T}$ by assumption and $\alpha_t^{n}$ converges to $0$ in probability. Moreover, $V_{t,T}^{c,1}1_{F_t}+V_{t,T}^{c,2}1_{\Omega\setminus F_t}\le V_{t,T}^{n}+\alpha_t^{n}$. Therefore, Proposition \ref{set-of-limits} implies that 
$V_{t,T}^{c,1}1_{F_t}+V_{t,T}^{c,2}1_{\Omega\setminus F_t}\in  \cV_{t,T}^c$ and the conclusion follows.
\end{proof}

\begin{lemm}\label{decomp} Let $h_T\in L^0(\R,\cF_T)$ be a payoff. If $\cV_{t,T}$  (resp. $\cV_{t,T}^c$) is $\cF_t$-decomposable (resp. infinitely  $\cF_t$-decomposable), then $\cP_{t,T}(h_T)$ (resp. $\cP_{t,T}^c(h_T)$) is $\cF_t$-decomposable (resp. infinitely  $\cF_t$-decomposable).
\end{lemm}
\begin{proof} Suppose that $\cV_{t,T}$ is $\cF_t$-decomposable and consider $p_t^1,p_t^2\in \cP_{t,T}(h_T)$ and $F_t\in \cF_t$. Then, $p_t^i+V_{t,T}^{i}\ge h_T$ for some $V_{t,T}^{i}\in \cV_{t,T}$, $i=1,2$. By assumption, we have $V_{t,T}=V_{t,T}^{1}1_{F_t}+V_{t,T}^{2}1_{\Omega\setminus F_t}\in \cV_{t,T}$ by assumption and $p_t^11_{F_t}+p_t^21_{\Omega\setminus F_t}+V_{t,T}\ge h_T$. We deduce that $p_t^11_{F_t}+p_t^2 1_{\Omega\setminus F_t}\in  \cP_{t,T}(h_T)$. By the same reasoning, the property holds for  $\cV_{t,T}^c$ and the infinite $\cF_t$-decomposability is obtained similarly.  The conclusion follows. 
\end{proof}

\begin{lemm} \label{PriceCriteria} Let $h_T\in L^0(\R,\cF_T)$ be a payoff. If $\cV_{t,T}$  is infinitely  $\cF_t$-decomposable, then for any $\gamma_t\in L^0(\R,\cF_t)$ such that $\gamma_t>\pi_{t,T}(h_T)$, there exists a price $p_t\in \cP_{t,T}(h_T)$ such that $p_t<\gamma_t$. In particular, $\gamma_t\in \cP_{t,T}(h_T)$. 
\end{lemm}
\begin{proof} Since $\cV_{t,T}$  is infinitely  $\cF_t$-decomposable, $\cP_{t,T}(h_T)$ is infinitely  $\cF_t$-decomposable by Lemma \ref{decomp}. Therefore, $\cP_{t,T}(h_T)$ is directed downward and we deduce that $\pi_{t,T}(h_T)\lim_n\downarrow p_t^n$ where $p_t^n\in \cP_{t,T}(h_T)$, see \cite[Section 5.3.1]{KS}. Then, a.s.($\omega$), there exits $n(\omega)$ such that $p_t^n(\omega)<\gamma_t(\omega)$. We then define 
\bean N_t&=&\inf\{n\ge 1:~p_t^n<\gamma_t\}\in L^0(\N,\cF_t),\\
p_t&=&\sum_{j=1}^\infty p_t^j1_{\{N_t=j\}}.\eean
By assumption $p_t\in \cP_{t,T}(h_T)$ and $p_t<\gamma_t$. The conclusion follows. \end{proof}

\end{appendices}


\bibliography{sn-bibliography}


\bibitem{BCL}
Baptiste J., Carassus L. and L\'epinette E. Pricing without martingale measure. Preprint, \url{https://hal.archives-ouvertes.fr/hal-01774150}, 2020.

\bibitem{BCJ} Barron E.N., Cardaliaguet P. and Jensen R. Conditional essential suprema with applications. Applied Mathematics and Optimization, 48, 229-253, 2003.

\bibitem{BS} Bayraktar E. and Sayit H. No arbitrage conditions for simple trading strategies. Annals of Finance, 6, 1, 147-156, 2010.

\bibitem{BSV} Bender C, Sottinen  T and Valkeila E. Pricing by hedging and no-arbitrage beyond semi-martingale. Finance and Stochastics, 12, 4, 441-468, 2008.

\bibitem{CL} Carassus L. and L\'epinette E.  Pricing without no-arbitrage condition in discrete-time.  Journal of Mathematical Analysis and Applications, 505, 1, 125441, 2021. 

\bibitem{Ch} Cheredito P. Arbitrage in fractional Brownian motion models. Finance and Stochastics, 7, 4, 533-553, 2003.
 
\bibitem{DMW}
Dalang, E.C., Morton, A. and Willinger, W., Equivalent martingale measures and no-arbitrage in stochastic securities market models. Stochastics and Stochastic Reports,  29, 185-201, 1990.

\bibitem{DScha1}
Delbaen, F. and Schachermayer, W., A general version of the fundamental theorem
of asset pricing. Mathematische Annalen, 300, 463-520, 1994.

\bibitem{EL} El Mansour M. and L\'epinette E. Conditional interior and conditional closure of a random sets. Journal of Optimization Theory and Applications, 187, 356-369, 2020.

\bibitem{FKV} Filipovic D, Kupper M. and Vogelpoth N. Separation and duality in locally $L^0$ convex modules. Journal of Functional Analysis,   256,  12,  3996-4029, 2009.

\bibitem{Gua1} Guasoni P. No arbitrage with transaction costs, with fractional Brownian motion and beyond. Mathematical Finance, 16, 2, 469-588, 2006.

\bibitem{JS} Jacod J. and Shiryaev A.N. Limit theorems for stochastic processes. Grundlehren der mathematischen Wissenschaften 288, a series of comprehensive studies in mathematics.  Springer-Verlag  Berlin, Heidelberg, 2003.

\bibitem{JPS} Jarrow R.,  Protter P. and  Sayit H. No arbitrage without semimartingales. Annals of Applied Probability, 19, 2, 596-616, 2009.

\bibitem{KSt}
Kabanov, Y. and Stricker, C. A Teachers' note on no-arbitrage criteria. In
S\'eminaire de Probabilit\'es, XXXV, volume 1755 of Lecture Notes in Math., Springer Berlin, 149-152, 2001.

\bibitem{KS}
Kabanov, Y. and Safarian, M., Markets with transaction costs. Mathematical Theory. Springer-Verlag, 2009.

\bibitem{KK} Karatzas I. and  Kardaras C. The num\'eraire portfolio in semimartingale financial models. Finance and Stochastics, 11, 447-493, 2007.

\bibitem{Kreps} Kreps D.M. Arbitrage and equilibrium in economies with infinitely many commodities. Journal of Mathematical Economics, 8, 15-35, 1981.

\bibitem{LM} L\'epinette E. and Molchanov I., Conditional cores and conditional convex hulls of random sets. Journal of Mathematical Analysis and Applications, 478, 2, 368-392, 2019.

\bibitem{Lo} Lo A.W. Long-term memory in stock market prices. Econometrica, 59, 5, 1279-1313, 1991.

\bibitem{Pakk} Pakkanen M. Stochastic integrals and conditional full support. Journal of Applied Probability, 47, 3, 650-667, 2010.

\bibitem{RW2009}	Rockafellar, R. T. and Wets R. J. B. Variational analysis, Grundlehren der mathematischen Wissenschaften, 317. Springer-Verlag Berlin Heidelberg, 1998.

\bibitem{Rogers} Rogers L.C.G. Arbitrage with fractional Brownian motion. Mathematical Finance, 7, 95-105, 1997.

\bibitem{Sayit} Sayit H. Absence of arbitrage in a general framework. Annals of Finance, 9, 611-624, 2013.

\bibitem{Sot} Sottinen T. Fractional Brownian motion, random walks and binary market models. Finance and Stochastics, 5, 3, 343-355, 2001.


\end{document}